\documentclass[10pt]{article}
\usepackage{dcolumn}
\usepackage{bm}
\usepackage{hyperref}
\usepackage{verbatim}       

\usepackage{color}

\usepackage{amsthm}
\usepackage{amssymb}

\usepackage{amstext}
\usepackage{psfrag}
\usepackage{amsmath}
\usepackage{amsfonts}
\usepackage{units}
\usepackage{mathrsfs}

\newcommand{\p}{\partial}

\newcommand{\dd}{{\rm d}}

\newtheorem{theorem}{Theorem}[section]
\newtheorem{corollary}[theorem]{Corollary}
\newtheorem{lemma}[theorem]{Lemma}
\newtheorem{proposition}[theorem]{Proposition}
\theoremstyle{definition}
\newtheorem{definition}[theorem]{Definition}
\theoremstyle{remark}
\newtheorem{remark}[theorem]{Remark}
\newtheorem{example}[theorem]{Example}

\begin{document}


\title{Totally geodesic null hypersurfaces and  constancy of surface gravity in  Finsler spacetimes}

\author{E. Minguzzi\footnote{Dipartimento di Matematica, Universit\`a degli Studi di Pisa,  Via
B. Pontecorvo 5,  I-56127 Pisa, Italy. E-mail:
ettore.minguzzi@unipi.it, ORCID:0000-0002-8293-3802}}

\date{}

\maketitle

\begin{abstract}
\noindent
We define and study totally geodesic null hypersurfaces in Finsler spacetimes. We prove that the null convergence condition and a certain mild gravitational equation $\chi_\alpha=0$, imply the vanishing of the restriction of the Ricci 1-form on the hypersurface. This makes it possible to extend to the Lorentz–Finsler setting essentially all notable results for compact totally geodesic null hypersurfaces that hold in the Lorentzian case. In fact, we introduce a trick that reduces the Lorentz–Finsler analysis to a purely Lorentzian study. As a result, it follows that, under the stated conditions, connected compact totally geodesic null hypersurfaces admit constant surface gravity. Further topological classification results are also obtained. The possibility of deriving these results from the dominant energy condition without using $\chi_\alpha=0$ is also explored, this strategy selecting some specific possibilities.
Since surface gravity can be interpreted as temperature in some contexts, and its constancy expresses the zeroth law of thermodynamics, the present work provides a compelling physical argument in favour of some special Finslerian gravitational equations.
\end{abstract}


\setcounter{secnumdepth}{2}
\setcounter{tocdepth}{2}
\tableofcontents

\section{Introduction}

This work is devoted to the study of totally geodesic null hypersurfaces in Finsler spacetimes, with a particular focus on compact ones. Our main aim is to establish the validity, within this framework, of several results already proven for null hypersurfaces in Lorentzian geometry. These results, especially those concerning the constancy of surface gravity (Corollary \ref{ocvt} below), carry deep physical significance; surface gravity has, in many contexts, particularly in black hole physics, been shown to admit a physical interpretation as temperature.

Some elements of causality theory will be required. Much of causality theory can be transferred to the Finsler spacetime setting \cite{minguzzi13d, javaloyes13, aazami14, minguzzi15}, the key step being the proof of the existence of convex neighborhoods and the local achronality property of null geodesics \cite{minguzzi13d}. Causality theory for Finsler spacetimes can also be derived from that for proper cone structures \cite{minguzzi17}. In this work, we will not need many elements of this theory, though some previous exposure of the reader on this topic may be beneficial. Our focus will be restricted to null hypersurfaces. There is limited material in this direction, aside from the relevant section in \cite{minguzzi15}, much of which will be reviewed below, and the study carried out in the context of cone structures in \cite{javaloyes22b}.

This work is motivated by physics, where there has been sustained and, in recent years, growing interest in possible Finslerian modifications of Einstein's theory of gravity \cite{pfeifer12,li14,minguzzi14c,fuster16,fuster18,hohmann19,caponio20}. Although several mathematical results can be generalized from Lorentz to Lorentz–Finsler geometry, there remains significant uncertainty regarding the identification of the gravitational field equations in the Finslerian setting. Of course, one should not merely attempt to guess the correct equations through algebraic analogies. Instead,  physical motivations are needed to support one set of equations over another.

In the absence of empirical evidence or predictions that could validate specific models, particular attention should be paid to mathematical constructions that admit clear physical interpretations. Our analysis will be of this type, establishing the constancy of surface gravity based solely on a few gravitational equations.

We recall that, physically, the constancy of surface gravity is related to the zeroth law of thermodynamics. This is well established in the context of Black Hole physics (in general relativity, hence standard Lorentzian geometry).
In fact, technically, the constancy of surface gravity is a kind of mild property of stationarity, implied by the Killing nature of the horizon \cite{chrusciel20} (whenever there is such symmetry). Naturally, if an asymptotically flat spacetime admits a  timelike Killing field  on the outer region, the evolution with emission of gravitational radiation must have stopped and so the horizon is expected to have reached constant temperature. This temperature  is then naturally identified with surface gravity (times $2\pi$), as it is also confirmed by the study of the Hawking radiation.
The interpretation extends to compact null hypersurfaces.
In fact, the study of compact null hypersurfaces has become one main tool to study stationary black holes (even in non-empty spacetimes) because, via a known  compactification trick, the stationary black hole horizon is reduced to a compact totally geodesic null hypersuface \cite{friedrich99,minguzzi24}.
The constant surface gravity of compact totally geodesic null hypersurfaces is thus indeed related to the zeroth law.

Returning to the Finslerian generalization, we shall follow two approaches, depending on the physical origin of the constancy of surface gravity. If it is derived from the null convergence condition ($Ric(v)\ge 0$ for $v$ future-directed lightlike) then the equation \(\chi_\alpha = 0\) is selected (even in presence of matter). If it is derived from the dominant energy condition then a family of possibility is selected (Thm.\ \ref{pcfq}). Two notable options single out. An equation with interesting conservation properties (\ref{cmxx}), already introduced in \cite{minguzzi14c}, and a novel  equation (\ref{cmbt})  which in vacuum is equivalent to $\chi_\alpha=0$ and $Ric=0$. The equation $Ric=0$ already enjoys broad consensus as it relates to the attractiveness of gravity via the Raychaudhuri equation \cite{rutz93}.


It is worth noting that other equations, such as \(H_{\alpha \beta} := \frac{1}{2} \left( \frac{\partial}{\partial v^\beta} \chi_\alpha + \frac{\partial}{\partial v^\alpha} \chi_\beta \right) = 0\), have useful algebraic consequences \cite{mo09,li15}, but only \(\chi_\alpha = 0\) is physically motivated by our argument. Interestingly, the equations $\chi_\alpha=0$ and $Ric=0$ have been shown to admit solutions in certain Finslerian generalizations of the Schwarzschild metric by Marcal and Shen \cite{marcal23}. These authors were led to consider $\chi_\alpha=0$ jointly with other possibilities due to its convenient algebraic implications.  As recalled below (Remark \ref{llrt}), \(\chi_\alpha = 0\) by implying $H_{\alpha \beta}=0$ reduces the indeterminacy of potential gravitational equations by inducing the coincidence of certain Ricci tensors \cite{mo09,sevim23}. Also it establishes the symmetry of some Ricci tensors.
Our argument provides  support from physical considerations and shows thus that it is indeed more than worthwhile to look at the consequences of this equation.

This work is structured as follows. In  Section \ref{pres} we recall some basic results of Lorentz-Finsler geometry fixing terminology and notation. We also introduce the Ricci 1-form, which does not seem to be much used in the literature but will prove of importance in this work. In Section \ref{ral} we define surface gravity for regular null hypersurfaces and introduce the shape operator $b$, instrumental for the subsequent section, and its Riccati equation.  In Section \ref{mvut} we provide several equivalent characterizations for totally geodesic null hypersurfaces and establish the equivalence between the vanishing of the restriction of the Ricci 1-form on a totally geodesic null hypersurface, $Ric_n\vert_{TH}=0$, and the equation $i_n\dd \omega=0$, see Eq.\ (\ref{noq}).
Then we prove that the null energy condition and $\chi_\alpha=0$ implies
$Ric_n\vert_{TH}=0$ over totally geodesic hypersurfaces, Thm.\ \ref{cmqq}.
The Section ends with the elaboration of a trick which allows us to embed the hypersurface into a Lorentzian manifold while preserving essentially all notions of geometrical and physical relevance including surface gravity and geodesic (in)completeness. This allows us to import result  to Lorentz-Finsler geometry from Lorentzian geometry. Section \ref{trasl} is devoted to the translation of the most relevant results, at least for this work perspective, though many others, in principle, could generalize. Finally, Section \ref{kxrt} reconsiders the possibility of deriving the key equation $Ric_n\vert_{TH}=0$ not from the null energy condition but from the dominant energy condition. This leads to the suggestion of a family of possible gravitational equations, one of them being compatible with a form of pointwise energy-momentum conservation.

\section{Preliminaries} \label{pres}

The purpose of this section is primarily to establish notation and terminology, but it may also serve as a brief introduction to Finsler geometry. In this sense we adapt  and streamline the presentation we gave in  \cite{minguzzi15}. As various schools have developed different conventions, we include key coordinate expressions to allow the reader to readily translate between notations. The objects introduced below can, of course, be formulated in a coordinate‑independent manner; for such treatments and for other introductions to Finsler geometry we refer to \cite{antonelli93,bao00,shen01,mo06,dahl06,szilasi14,minguzzi14c,ohta21}.

Let \( M \) be a paracompact, Hausdorff, connected, \( n+1 \)-dimensional manifold.\footnote{Subsequently the symbol “\( n \)” will also be used to denote a lightlike vector field; we trust this will not cause confusion.} Let \( \{x^\mu\} \) denote a local chart on \( M \) and \( \{ x^\mu,v^\nu\} \) the induced chart on \( TM \).
A Finsler Lagrangian is a function on the slit tangent bundle \( \mathscr{L}\colon TM\backslash 0 \to \mathbb{R} \) that is positively homogeneous of degree two in the velocities: \( \mathscr{L}(x,sv)=s^2 \mathscr{L}(x,v) \) for every \( s>0 \). The metric is defined as the Hessian of \( \mathscr{L} \) with respect to the velocities:
\begin{equation}
g_{\mu \nu}(x,v)= \frac{\partial^2 \mathscr{L}}{\partial v^\mu \partial v^\nu},
\end{equation}
and in index-free notation we shall also write \( g_v \) to emphasize the dependence on the velocity. This Finsler metric gives a map \( g\colon TM\backslash 0 \to  T^*M \otimes T^*M \).

Lorentz-Finsler geometry arises when \( g_v \) is Lorentzian, i.e.\ of signature $(-,+,$ $\dots,+)$ (in the positive-definite Finsler case one often works with a function \( F \) satisfying \( \mathscr{L}=F^2/2 \)).
The above definition of a Lorentz-Finsler manifold is due to John Beem \cite{beem70}. It is particularly convenient to work with a  Lagrangian defined on the whole slit bundle \( TM\backslash 0 \) because the theory of Finsler connections has traditionally been developed on this space.

The Lorentz-Finsler manifold is a Finsler spacetime if one of the connected components of the locus \( \mathscr{L}< 0 \) has been selected and termed {\em future timelike cone}. Its closure is the {\em future causal cone} while its boundary is the {\em future light cone}.\footnote{Actually we could drop completely the adjective {\em future} as we never really defined a past causal cone (and in this logic we could drop {\em future-directed} in all the instances that follow).  This would only be  meaningful  for reversible Finsler Lagrangians. Still we shall keep the adjective for analogy with general relativity.}

Although causality theory depends only on the restriction of the Lagrangian to the causal cone \( \mathscr{L}\le 0 \), such a restriction is not necessary; it does not add generality and it spoils the existence of convex neighborhoods. If one is given a Lagrangian defined only on the future causal cone, it is convenient first to extend it to the entire slit tangent bundle. This is possible whenever the Lagrangian and the light cone are sufficiently regular (see \cite{minguzzi14h} for a full discussion).



 For simplicity, we shall work with a smooth Lagrangian on \( TM\backslash 0 \). We could work with lower regularity. The consequence would simply be that the vector field $n$ accomplishing constant surface gravity would not be smooth anymore. We will not consider the problem of establishing a precise connection in the non-smooth case.

 In many examples of Lorentz-Finsler spaces the light cones are non-smooth. The results of this work  generalize to this case provided the lightlike semi-tangents to the null hypersurfaces do not belong to the non-smooth  directions of the light cone.

Let us recall some elements of pseudo-Finsler connections (the reader is referred to \cite{minguzzi14c}). The Finsler Lagrangian determines the geodesics as stationary points of the functional \( \int \mathscr{L}(x,\dot x)\,dt \). The Lagrange equations are of second order, and a good starting point for introducing Finsler connections is the notion of a {\em spray}.

A spray on \( M \) can be locally characterized by a second-order differential equation
\[
\ddot x^\alpha+2G^\alpha(x,\dot x)=0,
\]
where \( G^\alpha \) is positively homogeneous of degree two: \( G^\alpha(x,s v)=s^2 G^\alpha(x,v) \) for every \( s>0 \). Let \( E=TM\backslash 0 \) and \( \pi_M\colon E\to M \) be the usual projection. This projection determines a vertical space \( V_e E \) at each point \( e\in E \). A {\em non-linear connection} is a splitting of the tangent space \( TE=VE\oplus HE \) into vertical and horizontal bundles (for an introduction to non-linear connections see \cite{modugno91,michor08}). A basis for the horizontal space is
\[
\Big\{\frac{\delta \ }{\delta x^\mu}\Big\}, \qquad \frac{\delta}{\delta x^\mu}=\frac{\partial}{\partial x^\mu}-N^\nu_\mu(x,v) \frac{\partial}{\partial v^\nu},
\]
where the coefficients \( N^\nu_\mu(x,v) \) define the non-linear connection and transform appropriately under coordinate changes. The map $E\times TM \to TE$, $X^\mu\p/\p x^\mu \mapsto X^\mu\delta/\delta x^\mu$ is called {\em horizontal lift} (the first $E$ factor tells us where to evaluate the non-linear connection coefficient), similarly the map  $E\times TM \to TE$, $X^\mu\p/\p x^\mu \mapsto X^\mu\p/\p v^\mu$ is called {\em vertical lift}.

The curvature of the non-linear connection measures the non‑holonomicity of the horizontal distribution:
\begin{equation} \label{nsp}
\left[ \frac{\delta}{\delta x^\alpha}, \frac{\delta}{\delta x^\beta}\right]= -R^\mu{}_{\alpha \beta} \frac{\partial}{\partial v^\mu}, \qquad R^\mu{}_{\alpha \beta}=\frac{\delta N^\mu_\beta}{\delta x^\alpha}-\frac{\delta N^\mu_\alpha}{\delta x^\beta}.
\end{equation}
The non-linear curvature is thus a map $R^N: E \to TM\otimes  \Lambda^2 T^*M$, where $N$ here stands for {\em non-linear}. We shall also write it $R^N_v$ to stress its dependence on the support vector $v$. We shall have to deal with similar related objects soon. Note that when considering the components we drop $N$ as we can recognize which type of curvature we are considering from the number of indices.

We can define a {\em covariant derivative} for the non-linear connection as follows. Given a section \( s\colon U\to E \), \( U\subset M \),
\[
 D_{\xi} s^\alpha=\Big(\frac{\partial s^\alpha}{\partial x^\mu} +N^\alpha_\mu\big(x, s(x)\big) \Big)\xi^\mu .
\]
The {\em flipped or dynamical derivative}\footnote{Sometimes  called itself covariant derivative \cite{shen01}.} is instead
\[
{\tilde{D}_{\xi}} s^\alpha=\frac{\partial s^\alpha}{\partial x^\mu} \, \xi^\mu+N^\alpha_\mu(x, \xi) s^\mu .
\]
Although well defined, it is {\em not} a covariant derivative in the standard sense because it is non-linear in the derivative vector \( \xi \).
Observe that if \( X,Y\colon M\to TM \) are vector fields then
\begin{equation} \label{kkq}
{\tilde D_{X}} Y- D_Y X=[X,Y].
\end{equation}
A geodesic is a curve \( x(t) \) satisfying \( D_{\dot x} \dot{x}=0 \), or equivalently \( {\tilde D_{\dot x}} \dot x=0 \).
Unless otherwise stated, we only consider the non-linear connection determined by a spray via \( N^\mu_\alpha=G^\mu_\alpha:=\partial G^\mu/\partial v^\alpha \). The geodesics of this non-linear connection coincide with the integral curves of the spray. Moreover, the geodesics of the spray are the stationary points of the action functional \( \int\! \mathscr{L}\,dt \); consequently,
\begin{align}
2 {G}^\alpha(x,v)&= g^{\alpha\delta}\Big( \frac{\partial^2\mathscr{L} }{\partial x^\gamma \partial v^\delta} \,v^\gamma -\frac{\partial\mathscr{L} }{\partial x^\delta } \Big) \label{axo}\\
&=\frac{1}{2}\, g^{\alpha \delta} \Big( \frac{\partial}{\partial x^\beta} \,g_{\delta \gamma}+\frac{\partial}{\partial x^\gamma} \, g_{\delta \beta}-\frac{\partial}{\partial x^\delta}\, g_{\beta \gamma}\Big) v^\beta v^\gamma . \label{axu}
\end{align}


In Finsler geometry one can also define a {\em linear Finsler connection} \( \nabla \), i.e. a splitting of the vertical bundle \( \pi_E\colon VE\to E \), \( E=TM\backslash 0 \) (an equivalent pullback bundle approach could be used where $VE$ is replaced by $\pi^*TM$, where $\pi: E\to M$). The Berwald, Cartan, Chern–Rund and Hashiguchi connections are of this type; they are often called the {\em notable} Finsler connections. In this work we are going to completely avoid these connections in statements  because they will actually be needed only in a few proofs, that is, we follow the same philosophy of \cite{minguzzi15}. Indeed, the mathematics suggests that,  physically, the non-linear connection might be more important than the Finsler linear connection.



The notable linear Finsler connections are all compatible with the same non-linear connection. Indeed, the covariant derivative \( X\to \nabla_X L \) of the Liouville vector field \( L\colon E\to VE \), \( L=v^\alpha\partial/\partial v^\alpha \), vanishes precisely on an \( n+1 \)-dimensional distribution that determines a non-linear connection. This distribution is the same for all these connections and is obtained from the spray as described above.

Each Finsler connection \( \nabla \) determines two covariant derivatives, \( \nabla^H \) and \( \nabla^V \), obtained from \( \nabla_{\check X} \) when \( \check X \) is, respectively, the horizontal or vertical lift of a vector \( X\in TM \). In particular, \( \nabla^H \) is determined by local connection coefficients \( H^\alpha_{\mu \nu}(x,v):=e^\alpha(\nabla^H_{e_\nu} e_\mu) \), $e_\mu=\p/\p x^\mu$, which are related to those of the non-linear connection by (regularity) \( N^\alpha_\nu(x,v)=H^\alpha_{\mu \nu}(x,v)v^\mu \). Examples are the Berwald connection
\[
H^\alpha_{\mu \nu}:=G^\alpha_{\mu \nu}:=\frac{\partial}{\partial v^\nu}\,G^\alpha_\mu,
\]
or the Chern–Rund (or Cartan) connection, for both of which \( \nabla^H g=0 \) and hence
\begin{equation} \label{sog}
H^\alpha_{\beta \gamma}:=\Gamma_{\beta \gamma}^{\alpha}:=\frac{1}{2} g^{\alpha \sigma} \Big( \frac{\delta}{\delta x^\beta} \,g_{\sigma \gamma}+\frac{\delta}{\delta x^\gamma} \, g_{\sigma \beta}-\frac{\delta}{\delta x^\sigma}\, g_{\beta \gamma}\Big).
\end{equation}
The difference
\begin{equation} \label{lan}
L_{\beta \gamma}^{\alpha}=G_{\beta \gamma}^{\alpha}-\Gamma_{\beta \gamma}^{\alpha}
\end{equation}
is the {\em Landsberg tensor}. The tensor \( L_{\alpha \beta \gamma}(x,v)=g_{\alpha \mu}(x,v) L^{\mu}_{\beta \gamma}(x,v) \) is symmetric and satisfies \( L_{\alpha \beta \gamma}(x,v) v^\gamma=0 \).  The vertical covariant derivative $\nabla^V$ is determined by coefficients $V^\alpha_{\mu \nu}$. They vanish for the Berwald and Chern-Rund connections which are the covariant derivatives that will interest us.

When using \( \nabla^H \) one must specify at which point of \( TM \) the expression is evaluated; this point is often called the {\em support vector}. Observe that if \( X\colon E\to VE= E\times_M TM \) is a Finsler field whose components do not depend on \( v \), then at the support vector \( u \) we have \( \nabla^H_u X=\tilde D_u X \); while at the support vector \( X \), \( \nabla^H_u X=D_u X \). For this reason, whenever possible we employ \( D \) or \( \tilde D \) directly instead of \( \nabla^H \).

For instance, in \cite{minguzzi15} we noted that
a property of the flipped derivative, which follows from the horizontal compatibility of the Chern–Rund or Cartan connections with the metric, is
\begin{equation}
\tilde{D}_u g_u(X,Y)= g_u(\tilde{D}_u X,Y)+g_u(X,\tilde{D}_u Y),
\end{equation}
for every pregeodesic vector field \( u\in T_p M\backslash 0 \) and fields \( X,Y\colon M\to  TM \). The linearity of the map \( X\mapsto \tilde{D}_u X \) implies that \( \tilde{D}_u \) extends to one-forms and tensors in the usual way, so the previous equation is simply:
\begin{equation} \label{aao}
\tilde{D}_u g_u=0.
\end{equation}
Together with Theorem \ref{tiu}, the next proposition proved in \cite{minguzzi15} allows one to work directly with \( D \) or \( \tilde D \), reducing the need to refer to Finsler connections.
\begin{proposition} \label{cmbx}
For every \( u\colon U \to TM\backslash 0 \), \( U\subset M \), and \( X,Y\colon U\to  TM \),
\begin{align}
\frac{1}{2} \,\partial_X g_u(u,u)=g_u(u,D_Xu)&=\partial_u g_u(u,X)-g_u(\tilde D_u u,X)+g_u(u,[X,u]), \label{doi}\\
\partial_X g_u(u,Y)-\partial_Yg_u(u,X)&=g_u(Y,D_X u)-g_u(X,D_Yu)+g_u(u,[X,Y]). \label{doj}
\end{align}
\end{proposition}
Equation  (\ref{doi}) implies the well known fact that the non-linear parallel transport preserves the Finsler length of vectors.



The horizontal-horizontal curvature $R^{HH}$ of any Finsler connection can be defined  via the identity \cite{minguzzi14c}
\begin{align} \label{comc}
R^{HH}(X,Y) Z&=[\nabla^H_X \nabla^H_YZ-\nabla^H_Y \nabla^H_XZ-\nabla^H_{[X,Y]} Z]+\nabla^V_{R(X,Y)} Z,
\end{align}
 where $X,Y: M \to TM$, $Z: E \to E$. If $Z$ is projectable (its components do not depend on $v$) and the vertical connection coefficients vanish (e.g.\ for the Berwald or Chern-Rund connections) then the last term vanishes.

This $HH$-curvature is related to the curvature of the non-linear connection as follows (see e.g.\ \cite[Eq.\ (67)]{minguzzi14c})
\begin{equation} \label{cjmc}
R^{HH} {}^\alpha_{\ \beta \mu \nu}(x,v) v^\beta=R^\alpha{}_{\mu \nu}(x,v) ,
\end{equation}
equivalently, at support vector $v$,  $R^N_v(X,Y)=R^{HH}(X,Y) v$. The non-linear curvature satisfies for every $X,Y\in TM$
\begin{equation} \label{llop}
g_v(v,R^N_v(X,Y))=0,
\end{equation}
which follows from Eq.\ (\ref{cjmc}) if for Finsler connection one chooses the Cartan's one and notes the symmetry of $R^{HH}_{\textrm{Car}}{}_{\alpha \beta \mu\nu}$ in the first two indices  \cite[Eq.\ (82)]{minguzzi14c}.

Bao, Chern and Shen \cite{bao00,shen01}  work with just the Chern-Rund connection, and use a contracted tensor
\begin{equation} \label{hhn}
R^\alpha_{\ \beta}(x,v):= R^\alpha{}_{\beta \mu}(x,v) v^\mu=R_{\textrm{ChR}}^{HH}{}^\alpha{}_{\mu \beta \nu}(x,v) v^\mu v^\nu,
\end{equation}
where the Chern-Rund HH-curvature is
\begin{equation} \label{vll}
R^{HH}_{\textrm{ChR}}{}^\alpha{}_{\beta \gamma \delta}=\frac{\delta}{\delta x^\gamma} \,\Gamma^\alpha_{\beta \delta}-\frac{\delta}{\delta x^\delta} \, \Gamma^\alpha_{\beta \gamma}+\Gamma^\alpha_{\mu \gamma} \Gamma^\mu_{\beta \delta}-\Gamma^\alpha_{\mu \delta} \Gamma^\mu_{\beta \gamma}.
\end{equation}
The first identity in (\ref{hhn}) clarifies that the  endomorphism of $TM$ of components  $R^\alpha_{\ \beta}$ depends only on the curvature of the non-linear connection and not on the full Finsler connection.

Following usage by Bao, Chern and Shen \cite{bao00,shen01}  we shall denote it simply by \( R_v \), where the subscript \( v \) emphasizes the dependence on the point of \( E \). Note that the above equation can be rewritten, for any $X\in TM$,
\begin{equation}
R_v(X)=R^N_v(X, v).
\end{equation}
We define the {\em Ricci 1-form} at support vector $v$ as\footnote{It is a bit unfortunate, but we denote the Ricci 1-form and Ricci scalar similarly, $Ric_v$ and $Ric(v)$, respectively. The latter is also simply written $Ric$.}
\begin{equation}
Ric_v(Y)=\textrm{Tr}(X\to R^N_v(X, Y)), \quad (\textrm{in components} \ (Ric_v)_\alpha= R^\mu{}_{\mu \alpha}).
\end{equation}
The {\em Ricci scalar} is
\begin{equation}
 Ric(v):=\operatorname{tr} R_v =Ric_v(v).
 \end{equation}
 Observe that the Ricci 1-form and the Ricci scalar depend only on the non-linear connection and not on the Finsler linear connection (compare with expressions  such as $R^{HH}{}^\mu{}_{\alpha \mu \beta}$ which instead depend on the chosen linear connection).

The tensor \( R_v \) is related to the non-commutativity of the covariant and flipped derivatives as follows \cite{minguzzi15}.
\begin{theorem} \label{tiu}
Let \( u \) and \( X \) be fields on \( M \), and suppose \( u \) is  pregeodesic, i.e. \( D_uu=f u \) for some function \( f \). Then
\begin{equation} \label{com}
D_X \tilde{D}_u u-\tilde{D}_u D_X u-D_{[X,u]} u=R_u(X).
\end{equation}
\end{theorem}


As a corollary we obtain \cite{minguzzi15} \cite[Lemma 6.1.1]{shen01}
\begin{proposition}
Let \( x(t,s) \) be a geodesic variation, i.e. \( x_s:=x(\cdot, s) \) is a geodesic for each \( s\in (-\epsilon,\epsilon) \). Defining \( J=\partial  /\partial s\vert_{s=0} \), we have
\begin{equation} \label{cmo}
{\tilde D_{\dot x}}\,\,{\tilde D_{\dot x}} J+ R_{\dot{x}} (J)=0.
\end{equation}
\end{proposition}

\begin{proof}
Immediate from Eqs.\ (\ref{kkq}) and (\ref{com}) using \( [J,\dot{x}]=0 \) and \( D_{\dot{x}}\dot{x}=0 \).
\end{proof}

Note that the notion of a Jacobi field depends solely on the spray and its induced non-linear connection, not on the Finsler metric.


We also showed \cite{minguzzi15}

\begin{proposition} \label{xll}
The curvature endomorphism \( R_v \) of a Finsler Lagrangian \( \mathscr{L} \) satisfies for every \( X,Y\in T_pM \)
\begin{align}
R_v(v)&=0, \label{mki} \\
g_v(v, R_v(X))&=0, \label{mko}\\
g_v(X, R_v(Y))&=g_v(Y, R_v(X)). \label{mkp}
\end{align}
\end{proposition}
Thus it is symmetric. A similar property does not hold for the Ricci tensor if it is defined naively as $R^{HH}{}^\mu{}_{\alpha \mu \beta}$, at least if no other condition is imposed.
\subsection{Pseudo-Riemannian characterizations}

Let $s\colon M \to TM\backslash 0$ be a section, and let us consider the pullback
metric $s^*g$. Its components are $(s^*g)_{\alpha \beta}(x)=g_{\alpha
\beta}(x,s(x))=: (g_s)_{\alpha
\beta}(x)$. Let $\nabla^{s^*g}$ be the Levi-Civita connection of $s^*g$, and let
$\overset{s}{\nabla}$ be the pullback of a Finsler connection. Its coefficients are
\cite[Sect.\ 4.1.1]{minguzzi14c} \cite[Eq.\ (3.7)]{ingarden93}
\begin{equation} \label{jjd}
\big(\overset{s}{\nabla}\big)^\gamma_{\alpha \beta}=H^\gamma_{\alpha \beta}(x,s(x))+V^\gamma_{\beta \mu}(x,s(x))D_\alpha s^\mu.
\end{equation}
Particularly simple are the coefficients of the pullback
$\overset{s}{\nabla}{}^{ChR}$ of the Chern-Rund connection, since the vertical coefficients vanish. They are
$\Gamma^\gamma_{\alpha \beta}(x,s(x))$.

As the Chern-Rund connection is regular, $N^\alpha_\mu (x, s ) = \Gamma^\alpha_{ \nu \mu}(x,s) s^\nu$, we have, as recalled above,
\[
D_X s=\overset{s}{\nabla}{}^{\textrm{ChR}}_X s, \qquad \tilde D_s X=\overset{s}{\nabla}{}^{\textrm{ChR}}_s X .
\]
Let us recall the relationship between the Levi-Civita connection $\nabla^{s^* g}$ of the pullback metric and the pullback $\overset{s}{\nabla}{}^{\textrm{ChR}}$ of the Chern-Rund connection.
We have \cite[Prop.\ 1]{minguzzi15d}
\begin{proposition}
Let $s\colon M \to TM\backslash 0$ be a section. For every $X,Y\colon M \to TM$
\begin{equation} \label{dos}
\nabla^{s^* g}_X Y-\overset{s}{\nabla}{}^{\textrm{ChR}}_X Y=C(D_Ys,X)+C(D_Xs,Y)-g(C(X,Y),D
s)^\sharp.
\end{equation}
\end{proposition}

In the geodesic case the next result can be found in \cite[Sec.\ 6.2]{shen01} \cite[Lemma 4.3]{rademacher04} \cite[Lemma 2.3]{ohta15} \cite[Prop.\ 3.2]{lu19}.

\begin{corollary}[pseudo-Riemannian characterization of covariant derivatives]  \label{cotg} $\empty$\\
Let $s\colon M \to TM\backslash 0$ be a section. We have $D_s s=\tilde D_s s=\nabla^{s^*\! g}_s s=\overset{s}{\nabla}{}^{\textrm{ChR}}_s s$
thus the geodesic and pregeodesic conditions formulated by using the various connections are all equivalent.
If $s$ is pregeodesic then
\begin{align}
D_X s&=\nabla^{s^* \!g}_X s, \label{mnr1} \\
\tilde D_s X&=\nabla^{s^*\! g}_s X . \label{mnr2}
\end{align}
\end{corollary}

As a consequence, recalling Eq.\ (\ref{com}) we
obtain the following pregeodesic improvement of
\cite[Prop.\ 6.2.2]{shen01}
\begin{corollary}[pseudo-Riemannian characterization of the curvature endomorphism] $\empty$\\
If $s$ is pregeodesic then
\begin{align}
R_s(X)&=R^{s^* \!g} (X,s) s,\\
Ric(s)&=R^{s^* \!g}_{\alpha \beta} s^\alpha s^\beta. \label{clop}
\end{align}
\end{corollary}
\begin{proof}
Since $\tilde{D}_s s=f s$ for some function $f$, we have
\begin{align*}
D_X \tilde{D}_s s&=(\p_X f) s+f D_Xs=    (\p_X f) s+f \nabla^{s^* \!g}_X s=\nabla^{s^* \!g}_X \nabla^{s^*\! g}_s s, \\
\tilde{D}_s D_X s&=\tilde{D}_s \nabla^{s^* \!g}_X s= \nabla^{s^*\! g}_s \nabla^{s^* \!g}_X s, \\
D_{[X,s]} s&=\nabla^{s^*\! g}_{[X,s]} s,
\end{align*}
thus, using Eq.\ (\ref{com})
\[
R_s(X)=\nabla^{s^* \!g}_X \nabla^{s^*\! g}_s s-\nabla^{s^*\! g}_s \nabla^{s^* \!g}_X s-\nabla^{s^*\! g}_{[X,s]} s=R^{s^* \!g} (X,s) s .
\]
\end{proof}

\section{Regular null hypersurfaces} \label{ral}


The purpose of this section is to introduce the notion of surface gravity for regular null hypersurfaces, recall the quotient bundle formalism, introduce the null Weingarten map $b$ and the quotient curvature endomorphism $\bar R$.

In Lorentz-Finsler geometry a $C^1$  hypersurface $H$ is null if for every $p\in H$ there is a future-directed lightlike vector\footnote{Note that the letter $n$ enters also the spacetime dimension $n+1$, but this should not cause confusion.} $n\in T_p H$ such that $\textrm{ker} g_{n}(n,\cdot)= T_pH$. The  vector $n$ is unique up to positive  rescaling due to the assumed regularity of the light cone, see \cite{minguzzi13c} for more on the Legendre map in Lorentzian signature.

We recall the following generalization \cite{minguzzi15} of a standard result in Lorentzian geometry \cite{kupeli87,galloway00} which follows from the existence of convex neighborhoods and the local length maximization of lightlike geodesics \cite{minguzzi13d}.

\begin{theorem} \label{ppp}
Every $C^2$  hypersurface $H$ is null if and only if it is locally achronal and ruled by null geodesics.
\end{theorem}
The last terminology means that from every point of $H$ there passes a lightlike geodesic contained in $H$. These geodesics are also called {\em generators}.

It can  be shown  \cite{minguzzi15} that any $C^2$  null hypersurface $H$ (with boundary) can be enlarged to a hypersurface $H'$ by extending  the null generators in the future and past directions. The hypersurface will remain null as long as it remains $C^2$.

Let $H$ be a $C^2$ null hypersurface and let $n$ be a $C^1$ future-directed lightlike vector field  tangent to $H$ so that its integral curves are lightlike pregeodesics running over $H$. They satisfy
\begin{equation} \label{kap}
{\tilde D_{n}} n=D_n n=\kappa n
\end{equation}
where $\kappa$ is a function over $H$. Observe that $n$ is uniquely defined up to (point dependent) positive rescalings. The function $\kappa : H \to \mathbb{R}$ is called {\em surface gravity} and depends on the chosen field. In general $\kappa$ cannot be chosen to vanish without spoiling the regularity of $n$: this is due to the fact that the  generators can accumulate on themselves and that lightlike vectors do not have a privileged normalization.

The following result provides a convenient way to calculate the surface gravity \cite{moncrief83} \cite[Prop.\ 5]{minguzzi21} (recall that $\ker g_n(n,\cdot)=TH$)
\begin{proposition} \label{bkq}
Let $T$ be a vector field defined in a neighborhood of $H$ such that  $g_n(n,T)=-\frac{1}{a}=const\ne 0$ on $H$ (hence transverse to $H$), and extend $n$ to a neighborhood of $H$ in such a way that $L_Tn=0$, then
\begin{equation} \label{yry}
\kappa=\frac{a}{2} \p_T g_n(n,n) \vert_H.
\end{equation}
\end{proposition}

\begin{proof}
We have, using Prop.\ \ref{cmbx} and Eqs.\ (\ref{kkq}) and (\ref{aao})
\[
D_T g_n(n,n)=2 g_n(D_T n, n)=2 g_n(\tilde D_n T,n)=2[\p_n g_n(T,n)-g_n(T,\tilde D_n n)]=\frac{2\kappa}{a}  .
\]
\end{proof}

On the \(C^2\) null hypersurface we consider the vector bundle \(V = TH/\!\!\sim\), obtained by regarding two vectors \(X, Y \in T_pH\) as equivalent if \(Y - X \propto n\). Clearly, this bundle has \((n-1)\)-dimensional fibers. We denote by \(\overline{X}\) the equivalence class of \(\sim\) containing \(X\). At each \(p \in H\) we introduce:  a scalar product on $V_p$
\begin{equation} \label{cmpt}
h(\overline{X},\overline{Y}) := g_n(X,Y);
\end{equation}
 an endomorphism (shape operator, null Weingarten map)
  \[
  b : V_p \to V_p, \qquad \overline{X} \mapsto b(\overline{X}) := D_{\overline{X}} n := \overline{D_X n};
  \]
a second endomorphism
  \[
  \bar{R} : V_p \to V_p, \qquad \bar{R}(\overline{X}) := \overline{R_n(X)};
  \]
and a third endomorphism, which is the trace-free part of \(\bar{R}\),
  \[
  \bar{C} : V_p \to V_p, \qquad \bar{C}(\overline{X}) := \bar{R} - \frac{1}{n-1} \, \operatorname{tr} \bar{R} \; \text{Id}.
  \]

The definition of \(b\) is well posed because \(D_X\) is linear in \(X\) and \(D_n n \propto n\); hence \(D_{X+an} n = D_X n + k n\) for some \(k\). Moreover, by Prop.\ \ref{cmbx}, taking into account that $n$ is lightlike over $H$
\[
g_n(n, D_X n) =\frac{1}{2} D_X g_n(n,n) =0,
\]
which shows that \(D_X n \in TH\). The definition of \(\bar{R}\) is well posed because \(R_n(n)=0\) and, by Eq. (\ref{mko}), \(g_n(n, R_n(X))=0\); hence for every \(X \in T_pM\), \(R_n(X) \in T_pH\).

The endomorphisms \(b\), \(\bar{R}\) and \(\bar{C}\) are all self-adjoint with respect to \(h\), see \cite{minguzzi15}.  Moreover \cite{minguzzi15}
\begin{align}
\operatorname{tr} \bar{R} = \operatorname{tr} R_n = Ric(n). \label{sid}
\end{align}

\begin{definition}
The spacetime $(M,\mathscr{L})$ satisfies the  {\em null/timelike/causal convergence condition} if for  every future-directed lightlike/timelike/causal vector \(v\), \(Ric(v) \ge 0\).
\end{definition}

The derivative \(\tilde{D}_n\) induces a derivative \(\overline{X}' := \overline{\tilde{D}_n X}\) on sections of \(V\), and hence, in the usual way, a derivative on endomorphisms via \(E'(\overline{X}) := \bigl(E(\overline{X})\bigr)' - E(\overline{X}')\).

We arrive at the following Finslerian result \cite{minguzzi15} analog of a Lorentzian result  which can be found in \cite[Cor.\ 42]{kupeli87}.

\begin{proposition} \label{buu}
Along a generator of \(H\) the null Weingarten map satisfies the Riccati equation
\begin{equation} \label{nkw}
b' = -\bar{R} - b^2 + \kappa\, b.
\end{equation}
\end{proposition}
Additionally, we proved \cite{minguzzi15}

\begin{proposition} \label{bcu}
Along a generator of \(H\) the metric $h$ satisfies
\begin{equation}
h' = 0.
\end{equation}
\end{proposition}
%
%
%
%
%



Define the {\em expansion} \(\theta := \operatorname{tr} b\) and the {\em shear}
\[
\bar{\sigma} := b - \frac{1}{n-1}\, \theta \, \text{Id},
\]
so that \(\bar{\sigma}\) is the trace-free part of \(b\). For brevity write \(\sigma^2 := \operatorname{tr} \bar{\sigma}^2\). A trivial consequence of the definition is \(\sigma^2 \ge 0\), with equality if and only if \(\bar{\sigma} = 0\).

Taking the trace and the trace-free parts of Eq. (\ref{nkw}) we obtain

\begin{align}
\theta' &= -\operatorname{Ric}(n) - \sigma^2 - \frac{1}{n-1}\, \theta^2 + \kappa\, \theta, \qquad \text{(Raychaudhuri)} \label{raw} \\
\bar{\sigma}' &= -\bar{C} - \Bigl(\bar{\sigma}^2 - \frac{1}{n-1} \operatorname{tr} \bar{\sigma}^2 \; \text{Id}\Bigr) - \frac{2}{n-1}\, \theta\, \bar{\sigma} + \kappa\,\bar{\sigma}. \label{shw}
\end{align}

The term in parentheses is the trace-free part of \(\bar{\sigma}^2\); it vanishes in the physically relevant four-dimensional spacetime case (\(n=3\)).

\section{Totally geodesic null hypersurfaces} \label{mvut}
In this Section we define ``totally geodesic null hypersurfaces'' via the equation $b=0$. This naturally leads to a new object (Prop.\ \ref{cmop}), a 1-form $\omega$ over $H$ which happens to be connected with the restriction of the Ricci 1-form to the hypersurface (Thm.\ \ref{vie}). The equation $i_n \dd \omega=0$ (equiv. $Ric_n\vert_{TH}=0$) is derived from the null convergence condition and a certain equation $\chi=0$ (Thm. \ref{cmqq}). Once this is done many results can be obtained verbatim via a trick that embeds the hypersurface in a Lorentzian manifold (Thm.\ \ref{nner}).

To start with, observe that the Lie derivative $L_n$ makes sense over sections of the bundle $V\to H$. Indeed, if $\bar X$ is such section and $X$ is a representative
\[
L_n \bar X:= \overline{L_n X}
\]
where the expression on the right-hand side does not depend on the representative due to $L_n (\alpha n)=(L_n \alpha)n\propto n$.

\begin{proposition} \label{cmop}
Let $H$ be a $C^2$ null hypersurface and let $n$ be a lightlike tangent vector field on $H$. The following conditions are equivalent (and independent of $n$)
\begin{itemize}
\item[(i)] $b=0$,
\item[(ii)] $L_n \hat g=0$ where $\hat g=g\vert_{TH\times TH}$,
\item[(iii)] $L_n h=0$,
\item[(iv)] there is a 1-form $\omega: H \to T^*H$ such that $D_X n=\omega(X) n$ for every $X\in TH$.
\item[(v)] For a (not necessarily causal) extension of $n$ (and hence for every extension), we have for $X\in T_pH$, $p\in H$, $Y: H \to TH$,
    \[
    \nabla^{g_n}_X Y\in TH.
\]
\end{itemize}
If they hold we have on $H$, $\bar R=0$. In particular, $Ric(n)=0$ and $\bar C=0$.
\end{proposition}

The proofs are the same of the Lorentzian version provided the Levi-Civita connection $\nabla$ of Lorentzian geometry is replaced with $D$ or $\tilde D$ at suitable places.

\begin{proof}
The independence of $n$ is clear from (i) since if $n'=e^f n$ then $b'=e^f b$.

(i) $\Rightarrow$ (iv). The condition $b=0$ means $\overline{D_X n}=0$ that is $D_X n \propto n$. The linearity of the left-hand side in $X$ proves the existence of $\omega$.

(iv) $\Rightarrow$ (i). If $D_X n=\omega(X) n$ for every $X\in TH$ then $b=\overline{D_X n}=0$.

(i) $\Leftrightarrow$ (ii). The proof is as in \cite[Lemma 7]{minguzzi21}. Let $Y,Z: H \to TH$, so that $L_n Y, L_n Z:H \to TH$, and using Eq.\ (\ref{kkq}) and  (\ref{aao})
\begin{align*}
L_n &\hat g(X,Y)= L_n ( g_n(X,Y))-g_n(L_n X, Y)- g_n(X,L_n Y)\\
&=\p_n (g_n(X,Y))- g_n(\tilde D_n X, Y)-g_n(X,\tilde D_n Y)+ g_n( D_X n, Y)+g_n(X, D_Y n)\\
&=\tilde D_n g_n(X,Y)+ g_n( D_X n , Y)+g_n(X,D_Y n)= g_n( D_X n, Y)+g_n(X, D_Y n)\\
&=h(b(\bar X), \bar Y)+h(\bar X, b(\bar Y))=2 h(b(\bar X), \bar Y),
\end{align*}
where in the last step we used  that $b$ is self-adjoint with respect to $h$ \cite{minguzzi15}. Since $h$ is non-degenerate the right-hand side vanishes iff $b$ vanishes.

(ii) $\Leftrightarrow$ (iii). This follows from the first line of the previous equation in display where the right-hand side can also be rewritten $L_n (h(\bar X,\bar Y))-h(\overline{L_n X}, \bar Y)-h(\bar X,\overline{L_n Y})=L_n h (\bar X,\bar Y)$.

(i) $\Leftrightarrow$ (v). Let $Y:H\to TH$ be a vector field on $H$, and $X\in T_pH$, $p\in H$. Then, for any extension $n$, using Cor.\ \ref{cotg},
    \begin{align*}
    g_n(n,\nabla^{g_n}_X Y)= \nabla^{g_n}_X(g_n(n, Y))-g_n(\nabla^{g_n}_X n,Y)=-g_n(D_X n,Y)=-h(b(\bar X),\bar Y),
    \end{align*}
    thus $b=0$ iff $\nabla^{g_n}_X Y\in TH$.

The last statement follows from the Riccati equation (\ref{nkw}).
\end{proof}

\begin{definition}
A $C^2$ null hypersurface is said to be {\em totally geodesic} if any of the equivalent conditions of Prop.\  \ref{cmop} hold.
\end{definition}

\noindent {\em Warning.} In spite of the terminology, we did not prove that this property is equivalent to: every geodesic starting tangent to $H$ remains in $H$. In fact, we do not expect this  to be an equivalent characterization. What we proved, via characterization (v), is that $H$ is totally geodesic iff so is for a connection $\nabla^{g_n}$ where there is no need to specify the field extension. We shall return on this important point later on. Observe that in the totally geodesic case the connection induced by $\nabla^{g_n}$ on $TH\to H$ depends on the extension, however, the connection induced by it on the quotient bundle $V=TH/n\to H$ does not depend on the extension because, by \cite[Prop.\ 4]{minguzzi24c}, it is the transverse Levi-Civita connection, a connection which only depends on the quotient metric $h$ (and so not on $n$).

We stress that there seem to be many inequivalent definitions of totally geodesic hypersurface in Finsler geometry \cite{matsumoto85,berk09,bejancu14}. Our definition is better suited for our purposes and seems different from that of previous approaches.
\\

The surface gravity admits the following expression
\[
\kappa=\omega(n).
\]
Both $\omega$ and $\kappa$ depend on the choice of lightlike field $n$. From property (iv) we get that if  $n'= e^f n$, then $\omega'=\omega+\dd f$. Indeed,
\[
\omega'(X) e^fn=D_X (e^f n)=e^f X(f) n+e^f D_X n=[\p_X  f+\omega(X) ] e^f n.
\]
Let $\Lambda(p)\in (0,\infty]$ be the affine length of the geodesic $\gamma$ such that $\gamma(0)=p$, $\dot\gamma(0)=n$. The rescaling of $n$ (gauge transformation) is thus connected with the following contextual changes \cite{moncrief20,reiris21,minguzzi21}
\begin{align}
n'&=e^f n, \label{ok2} \\
\omega'&=\omega + \dd f, \label{ok3}\\
\Lambda'&=e^{-f}\Lambda,  \label{ok4} \\
\kappa'&= e^f(\kappa + \p_n f). \label{ok5}
\end{align}



The following is a Finslerian generalization of  \cite[Lemma 3]{minguzzi21}.\footnote{It should be noted that the statement concerning the bilinear form $\mu$ introduced in that Lemma cannot be generalized due to the replacement of the usual curvature symmetry with a more complicated identity (see the second equation in display in Sec.\ 5.4.1 of \cite{minguzzi14c}). 
}

\begin{theorem} \label{vie}
Let $H$ be a totally geodesic null hypersurface. For vector fields $X,Y\in TH$
\begin{align}
R_n^N(X,Y)&=\dd \omega(X,Y) n, \\
Ric_n(Y)&=\dd \omega(n,Y) . \label{noq}
\end{align}
\end{theorem}

\begin{proof}
From Eq.\ (\ref{cjmc}) and  equation  (3.11) of \cite{ingarden93}  we have at support vector $n$,
\[
R_n^N(X,Y)=R^{HH}_{\textrm{ChR}}(X,Y)n=\overset{n}{R}(X,Y)n-L(D_Xn,Y)+L(D_Yn,X) ,
\] where $\overset{n}{R}$ is the curvature of the pullback Chern-Rund connection $\overset{n}{\nabla}$ i.e.\ that of coefficients $\Gamma^\alpha_{\beta \gamma}(x,n(x))$. The last two terms vanish because $D_Xn, D_Yn\propto n$ and the Landsberg tensor in annihilated by the support vector.
Noting that $D_Xn=\overset{n}{\nabla}_X n$ and similarly for $Y$
\begin{align*} 
\overset{n}{R}(X,Y)n&=\overset{n}{\nabla}_X\overset{n}{\nabla}_Yn-\overset{n}{\nabla}_Y\overset{n}{\nabla}_Xn-\overset{n}{\nabla}_{[X,Y]}n \\
&=\overset{n}{\nabla}_XD_Yn-\overset{n}{\nabla}_YD_Xn-D_{[X,Y]}n\\ &=\overset{n}{\nabla}_X[\omega(Y)n]-\overset{n}{\nabla}_Y[\omega(X) n]-\omega([X,Y])n\\
&=\{X[\omega(Y)] n+\omega(Y) D_Xn\}-\{Y[\omega(X)]n +\omega(X) D_Yn\}-\omega([X,Y])n\\
&= \{X[\omega(Y)]-Y[\omega(X)]-\omega([X,Y])\}n=[\dd \omega(X,Y)] n .
\end{align*}

Let $N$ be a $g_n$-lightlike vector field on $H$ such that  $g_n(N,n)=-1$.
Let $e_1,\cdots,$ $ e_{n-1}\in T_pH$ be such that $\textrm{Span}(e_1, \cdots, e_{n-1})=\textrm{ker}g_n(N,\cdot)\cap TH$.
Note that on $TM$ the basis dual to $(N,n,e_1, \cdots, e_{n-1})$ is $( -g_n(n,\cdot),-g_n(N,\cdot), e^1, $ $\cdots e^{n-1})$ where $e^i(e_j)=\delta^i_{j}$, $e^i(n)=e^i(N)=0$, $i,j=1,\cdots, n-1$.
On $TH$ the basis dual to $(n,e_1,\cdots, e_{n-1}))$ is $(-g_n(N,\cdot),  e^1,\cdots, e^{n-1}))$.

The Ricci 1-form is
\[
Ric_n(Y)=\textrm{Tr}(X\to R^N_n(X, Y))=\sum_a b^a (R^N_n( b_a, Y))
\]
where $\{b_a, a=0,1,\cdots,n\}$ is any basis with $\{b^a\}$  the dual basis. In particular, we can let $\{b_a\}=\{N,n,e_1, \cdots, e_{n-1}\}$.
We must be careful in calculating the Ricci tensors, since the traces are  different in $M$ or $H$. Actually, this is not  a problem because in the   calculation of $Ric_n(Y)$ there appears the term
\[
b^0 (R^N_n(b_0, Y) )= -g_n(n, R^N_n( N, Y) )=0
\]
which vanishes due to Eq.\ (\ref{llop}). We conclude that for the calculation of $Ric_n(Y)$ we can use the trace restricted to $TH$ and then for $Y\in TH$
\[
Ric_n(Y) =\textrm{Tr}\{X\mapsto R^N_n(X,Y)\} = \textrm{Tr}\{X\mapsto \dd \omega(X,Y) n\}=\dd \omega(n,Y).
\]
\end{proof}

\begin{corollary} \label{vib}
For a totally geodesic null hypersurface the 1-form $Ric_n\vert_{TH}$ annihilates $n$ and so passes to the quotient to a 1-form $\overline{Ric_n} : H\to V^*$.
\end{corollary}

\begin{proof}
From Prop.\ \ref{cmop} $Ric(n)=Ric_n(n)=0$ and the claim follows.
\end{proof}

\begin{proposition}
For a totally geodesic null hypersurface  the following conditions are all equivalent
\begin{itemize}
\item[(i)] $Ric_n\vert_{TH}=0$,
\item[(ii)] $Ric_n(\cdot) \propto g_n(n,\cdot)$ on $H$,
\item[(iii)] $\overline{Ric_n}=0$.
\end{itemize}
\end{proposition}

\begin{proof}
(i) $\Rightarrow$ (ii). If $Ric_n\vert_{TH}=0$ then $\ker Ric_n\supset TH=\ker g_n(n,\cdot)$ thus $Ric_n(\cdot) \propto g_n(n,\cdot)$.
(ii) $\Rightarrow$ (i). If $X\in TH$, then $g_n(n,X)=0$ hence $Ric_n(X)=0$.
(iii) $\Rightarrow$ (i). For each $X\in TH$, $Ric_n(X)=\overline{Ric_n}(\bar X)=0$.
(i) $\Rightarrow$ (iii). For each $\bar X\in V$ let $X\in TH$ be a representative, then $\overline{Ric_n}(\bar X)=Ric_n(X)=0$
\end{proof}

The following result has no Lorentzian counterpart

\begin{proposition} \label{cnph}
 On a totally geodesic null hypersurface $H$, let $n$ be, as usual, a non-vanishing future-directed lightlike vector field tangent to it. We have at support vector $n$ the identity $R^\alpha{}_\beta A^\beta{}_{\alpha \mu}=0$  over $H$ where $R^\alpha{}_\beta$ is the curvature endomorphism and $A_{\alpha \beta \gamma}$ is any tensor  annihilated by the Liouville vector field  in the first two indices ($A_{\alpha \beta \gamma} v^\alpha=A_{\alpha \beta \gamma} v^\beta=0$) e.g.\ the Cartan torsion $C_{\alpha \beta \gamma}$ or the Landsberg tensor  $L_{\alpha \beta \gamma}$.

Furthermore,    $R^\alpha{}_\beta B_{\alpha}\propto n_\beta$  where $B$ is any 1-form field annihilated by the Liouville vector field, e.g.\ the mean Cartan torsion $I_{\alpha}$ or the mean Landsberg tensor  $J_{\alpha}$.
\end{proposition}

\begin{proof}
Let $p\in H$ and let us  introduce on $T_pM$ and its dual the same basis and cobasis introduced in the proof of Thm.\ \ref{vie}. We shall assign to them the labels $+,-,1,2,\cdots, n-1$. So, for instance, $e_+=N$, $e^+=-g_n(n,\cdot)$, $e_-=n$, $e^-=-g_n(N,\cdot)$.

We have the following component expressions at support vector $n$.
Equation (\ref{mki}) reads $R^\alpha{}_-=0$. Equation (\ref{mko}) reads $R^+{}_\beta=0$. The identity $A^\beta{}_{\alpha \mu} n^\alpha=0$ reads $A^\beta{}_{- \mu}=0$. The identity $A^\beta{}_{\alpha \mu} n_\beta=0$ reads $A^+{}_{\alpha \mu}=0$. Let us focus on the non-vanishing contributions for  $R^\alpha{}_\beta A^\beta{}_{\alpha \mu}$. Here $\alpha$ can't be - due to $A^\beta{}_{- \mu}=0$ but can't be + due to  $R^+{}_\beta=0$. Moreover, $\beta$ can't be - due to $R^\alpha{}_-=0$ and can't be + due to $A^+{}_{\alpha \mu}=0$. Thus all non-vanishing contribution come from $R^i{}_j A^j{}_{i \mu}$ for $i,j=1,\ldots, n-1$.
However, Prop.\ \ref{cmop} gives $\bar R=0$ which means $R_n(X) \propto n$ for every $X\in TH$. This implies $R^i{}_j=0$ and we have finished.

For the second statement, as noted we have: $R^\alpha{}_-=R^+{}_\beta=R^i{}_j=0$, $B_-=0$. Thus $R^\alpha{}_\beta B_{\alpha} e^\beta=R^\alpha{}_+ B_\alpha e^+\propto g_n(n,\cdot)$.
\end{proof}

Now, we are going to show that the condition $Ric_n\vert_{TH}=0$ can be implied by the null convergence condition, at least in some Finslerian gravitational theories.
This conclusion  depends on some mild additional conditions that ideally come  from  the gravitational dynamical equation imposed.

\begin{proposition} \label{kkj}
Let $N^\alpha_\beta$ be a torsionless ($\frac{\p }{\p v^\gamma} N^\alpha_\beta=\frac{\p }{\p v^\beta} N^\alpha_\gamma$) non-linear connection positive homogeneous of degree one (e.g.\ $N^\alpha_\beta=\frac{\p}{\p v^\beta} G^\alpha$ where $G^\alpha$ a spray, possibly induced by a pseudo-Finsler Lagrangian \cite[Eq.\ (36)]{minguzzi14c}).  Let $R^\alpha{}_{\beta \gamma}$ be its non-linear curvature, then the equation
\begin{align} \label{one}
-2\chi_\alpha:=\Big(\frac{\p}{\p v^{\gamma}} R^\gamma{}_{\alpha\nu}\Big) v^\nu&= 0,
\end{align}
is equivalent to
\begin{align}\label{two}
\frac{\p }{\p v^\alpha} Ric(v)&= 2 (Ric_v)_\alpha.
\end{align}
In the context of pseudo-Finsler geometry, it is also equivalent to
\begin{equation} \label{three}
 \frac{1}{2} \frac{\p^2 }{\p v^\alpha \p v^\beta} Ric(v)=  R^{HH}_{Ber} {}^\mu{}_{\alpha \mu \beta},
\end{equation}
where $R^{HH}_{Ber} $ is the horizontal-horizontal  curvature tensor of the Berwald connection.
\end{proposition}

\begin{remark}
Note that the affine sphere condition $I_\alpha=0$, defining the context of affine sphere spacetimes  \cite{minguzzi15e}, implies the symmetry of the Berwald Ricci tensor \cite[Thm.\ 5.2]{minguzzi14c}, $R^{HH}_{Ber} {}^\mu{}_{\alpha \mu \beta}=R^{HH}_{Ber} {}^\mu{}_{\beta \mu \alpha}$, and hence, by the Bianchi identities \cite[Eq.\ (73)]{minguzzi14c}, $R^{HH}_{Ber} {}^\mu{}_{ \mu \alpha \beta} v^\beta=0$, namely Eq.\ (\ref{one}).
\end{remark}

\begin{proof}
By the contracted Bianchi identities \cite[Eq.\ (64)]{minguzzi14c}  $\frac{\p}{\p v^{\alpha}} R^\gamma{}_{\beta \gamma}+\frac{\p}{\p v^{\beta}} R^\gamma{}_{\gamma \alpha}+\frac{\p}{\p v^{\gamma}} R^\gamma{}_{\alpha\beta }=0$
we get, contracting with $v^\alpha$ and using the positive homogeneity of degree one of the non-linear curvature,
\[
 R^\gamma{}_{\beta \gamma} +\Big(\frac{\p}{\p v^{\beta}} Ric(v) -R^\gamma{}_{\gamma \beta}\Big)-\Big(\frac{\p}{\p v^{\gamma}} R^\gamma{}_{\beta \alpha}\Big) v^\alpha=0,
\]
which reads
\begin{equation} \label{cnng}
-2\chi_\alpha=\frac{\p}{\p v^{\alpha}} Ric(v) -2R^\gamma{}_{\gamma \alpha},
\end{equation}
from which the equivalence of Eqs.\ (\ref{one}) and (\ref{two}) is clear. In pseudo-Finsler geometry it is well known that  \cite[Eq.\ (69)]{minguzzi14c}  $\frac{\p}{\p v^{\beta}} R^\alpha{}_{\gamma \nu}=R^{HH}_{Ber} {}^\alpha{}_{\beta\gamma \nu}$, thus  (\ref{two}) implies  (\ref{three}) by differentiation. It is also well known that \cite[Eq.\ (67)]{minguzzi14c}  $R^{HH}_{Ber} {}^\alpha{}_{\beta\gamma \nu}  v^\beta= R^\alpha{}_{\gamma \nu}$, thus  (\ref{three}) implies (\ref{two}) contracting with $v^\alpha$ and using positive homogeneity.
\end{proof}

In Finsler geometry the quantity $\chi_\alpha$ was studied in \cite{li15} (see also \cite{mo09}) where it was proved that  the symmetry  $R^{HH}_{Ber} {}^\mu{}_{\alpha \mu \beta}=R^{HH}_{Ber} {}^\mu{}_{\beta \mu \alpha}$, which by the first Bianchi identity is equivalent to $R^{HH}_{Ber} {}^\mu{}_{\mu \alpha \beta}=0$,  is also equivalent to $\chi_\alpha=0$ \cite[Thm.\ 1.1]{li15}. In fact, the following formulas hold
\[
\chi_\alpha=-\frac{1}{2}R^{HH}_{Ber} {}^\mu{}_{\mu \alpha \nu}v^\nu, \qquad \frac{\p \chi_\beta}{\p v^\alpha}-\frac{\p \chi_\alpha}{\p v^\beta}= R^{HH}_{Ber}{}^\mu{}_{\mu\alpha \beta}= R^{HH}_{Ber}{}^\mu{}_{\alpha\mu \beta}- R^{HH}_{Ber}{}^\mu{}_{\beta \mu \alpha}.
\]
The form of the identities of Prop.\ \ref{kkj} that will prove most useful for us is  (\ref{two}).

The content of this work will make a strong case in support of this equation (jointly with the null convergence condition and hence $Ric=0$ in the vacuum case), as a  Finslerian gravitational equation. In fact, thanks to it, we shall be able to reproduce many results of physical relevance such as the constancy of temperature (surface gravity) over compact horizons  in analogy to the standard Lorentzian theory. A key step is provided by the following theorem

\begin{theorem} \label{cmqq}
Suppose that $(M,\mathscr{L})$ satisfies the null convergence condition and Eq.\  (\ref{two}). Let $H$ be a null hypersurface and let $n$ be a tangent future-directed lightlike field. Suppose $Ric(n)=0$ on $H$ (e.g.\ because $H$ is totally geodesic or $(M,\mathscr{L})$ is Ricci flat) then  we have $Ric_n\vert_{TH}=0$.
\end{theorem}

\begin{remark}
We recall that, by Theorem \ref{vie}, for a totally geodesic hypersurface $H$, the equation $Ric_n\vert_{TH}=0$ is equivalent to
\begin{equation}
i_n\dd \omega=0.
\end{equation}
\end{remark}

The equation $Ric_n\vert_{TH}=0$  is crucially related to  the machinery of the ribbon argument and essential to prove the dichotomy between degenerate and non-degenerate horizons, among other results. The reader will see little of these arguments and tools in what follows as we shall be able to sidestep most lengthy proofs by reducing the problems to their Lorentzian, hence already established, versions.

The theorem implies that the 2-form $\dd \omega$, section of $T^*H\otimes T^*H$, passes to the quotient to a section of $V^*\otimes V^*$, $V=TH/n$. The quotient space $V$ is  also endowed with the scalar product $h$ which is the quotient of $g$.

\begin{proof}
Let $p\in H$ and let us denote with $E_p\subset T_pM\backslash 0$ the set of future-directed lightlike vectors at $p$. Let $\mathscr{L}_p:=\mathscr{L}\vert_{T_pM\backslash 0}$,  $Ric_p:=Ric\vert_{T_pM\backslash 0}$.
At $p\in H$ we have
$Ric_p(n)=0$ so the function $Ric_p$ reaches a minimum on $E_p$ at $n$. What follows is a kind of Lagrange multiplier argument. The tangent space to $E_p$ at $n$, namely $\ker  \dd \mathscr{L}_p \vert_n= \ker g_n(n,\cdot)$ (note that  $g_n(n,\cdot)$ does not vanish  because $g_n$ is Lorentzian, hence non-degenerate, and $n\ne 0$, because it is lightlike) must be included in $\ker \dd Ric_p\vert_{v=n}$ (otherwise it would be possible to reach positive and negative values of $Ric_p$ in a neighborhood of $n$ while remaining in $E_p$) which implies $ \dd Ric_p\vert_{v=n} \propto g_n(n,\cdot)$ or using Eq.\ (\ref{two}) $2(Ric_n)_\alpha=\frac{\p}{\p v^\alpha }Ric \vert_{v=n} \propto  (g_n)_{\alpha \beta} n^\beta$. Evaluating on vectors in $TH$  gives the desired result.
\end{proof}

\begin{remark} \label{llrt}
Another, less physically compelling, but still interesting  argument in favour of (\ref{two}) passes through the observation of its algebraic convenience. Via the equivalent equation (\ref{three}) some proposals for Finslerian gravitational equations that have been proposed in the literature actually coincide. Indeed, most of them have been written, following an algebraic analogy, in the form
\begin{quote}
{\em Ricci tensor} - $\frac{1}{2}$ {\em  curvature scalar} $\times$  {\em metric} = {\em  source},
\end{quote}
since the scalar curvature is the contraction of the Ricci tensor with the metric,  some arbitrariness is resolved if the Ricci tensor is identified (there is still the ambiguity in the choice of Finsler connection, not to mention the possibility of adopting these equations in a Sasakian rather than Finslerian framework).
For the Berwald connection, Eq.\  (\ref{three}) shows that two possible choices (the Akbar-Zadeh choice on the left, and the standard one on the right) coincide (and also the latter is thus symmetric) \cite[Prop.\ 3.4]{mo09}\cite{sevim23}. There is also a third choice, motivated by energy-momentum conservation, that we   introduced in \cite[Sec.\ 5.4.2]{minguzzi14c} (denoted with $\mathring{Ric}$ in\footnote{The first two indices of our curvature are switched with respect to the conventions of that work as we use the convention of \cite{misner73}.}  \cite{sevim23}) which, unfortunately, does not coincide with the other two.
\end{remark}

Using Cor.\ \ref{cotg} we prove the following key result

\begin{theorem}
Let $H$ be a $C^2$ null hypersurface on the spacetime $(M,\mathscr{L})$.
Let $T$ be a $g_n$-timelike vector field on $H$ such that $g_n(n,T)<0$.
Extend $n$ and $T$ to  vector fields in a neighborhood $U$ of $H$,
then $H$ is a $C^2$ null  hypersurface on the spacetime $(U,g_n)$ where the time orientation is provided\footnote{As on $H$, $g_n(T,T)<0$ and $g_n(n,T)<0$, these inequalities hold by continuity in a neighborhood of $H$. In particular, the continuity of $T$ allows to define a time-orientation compatible with that of $n$ at $H$. Observe that we are not assuming that $n$ is causal in a neighborhood of $H$.} by $T$.



Hypersurface $H$ is a totally geodesic on  $(M,\mathscr{L})$ iff it is   totally geodesic on
  $(U,g_n)$.
 In fact, the induced $\omega$ and hence the surface gravity is the same in the two interpretations because they share the same 1-form $\omega$ as
\begin{equation} \label{lppo}
D_X n=\nabla^{g_n}_X n =\omega(X) n.
\end{equation}
\end{theorem}

%

Let $\gamma : I \to H$, $t\mapsto \gamma(t)$ be such that  $\dot \gamma=\alpha n$, $\alpha(t)>0$ and so have image in the integral curves of $n$. The equation (\ref{mnr1}) also implies
\[
\nabla^{g_n}_{\dot \gamma} \dot \gamma=\nabla^{g_{\dot \gamma}}_{\dot \gamma} \dot \gamma=D_{\dot \gamma} \dot \gamma.
\]
Thus $\gamma(t)$ is a parametrized lightlike geodesic for $(U,g_n)$ iff it is a parametrized geodesic for $(M,\mathscr{L})$. This implies that the function $\Lambda(p)$ has the same interpretation in terms of affine length of the geodesic starting with tangent $n$ at $p$ in both spacetimes  $(U,g_n)$ and $(M,\mathscr{L})$.


\begin{proof}
The equation $g_n(n,n)=0$ at $H$ implies that the tangent $n$ to $H$ is lightlike according to the spacetime  $(U,g_n)$. The equation $TH=\ker g_n(n,\cdot)$ implies that $TH$ is a null hypersurface in $(U,g_n)$.
Equation (\ref{lppo}) is immediate from Cor.\ \ref{cotg}.

Characterization (iv) of Proposition \ref{cmop} read in the Finslerian version for $(M,\mathscr{L})$, or in the restricted standard Lorentzian version for $(U,g_n)$ shows that $H$ is totally geodesic in one spacetime if so is in the other, where $\omega$ is the same in both interpretations. Alternatively, use characterization (v) of the same Proposition.
\end{proof}

It is important to note that in Theorem \ref{cmqq} we  proved the equation $i_n\dd \omega=0$ via the Finslerian null convergence condition and Eq.\  (\ref{two}), not via the Lorentzian null convergence condition $g_n(X,X)=0, g_n(n,X)\le 0  \Rightarrow Ric^{g_n}(n,n)\ge 0$ which might not hold (though the Ricci scalars are the same, see Eq.\ (\ref{clop})).

The main idea is expressed by the following broad consequence
\begin{theorem} \label{nner}
 Results for the null hypersurface $H$ in $(U,g_n)$, when expressed in terms of the totally geodesic property, the metric $g_n$, the quotient metric $h$, the vector field $n$, the 1-form $\omega$, the surface gravity $\kappa$, the parametrization of geodesics hence the function $\Lambda$, the Ricci scalar $Ric(n)$ and curvature endomorphism $R_n$ on $H$, the topological properties of the flow on $H$, pass through to $H$ as null hypersurface of $(M,\mathscr{L})$.
\end{theorem}

or

\begin{theorem} \label{nndr}
Let $H$ be a totally geodesic null hypersurface on $(M,\mathscr{L})$ and suppose
$Ric_n\vert_{TH}=0$ (e.g.\ because   the null convergence condition and Eq.\  (\ref{two}) hold).
 To this Lorentz-Finsler case
 apply the same results that would apply in Lorentzian geometry provided they  involve the concepts mentioned in Thm.\ \ref{nner}.
\end{theorem}

This leads to the immediate translation of plenty of results proved in the traditional Lorentzian theory to the Lorentz-Finsler theory.

In some cases these results could have been proved following step by step the original Lorentzian proofs and looking for  some (often non-trivial) adaptation of the arguments. It is true, for instance, that the dichotomy or the topological results on the flow, proved below, could have been approached in this way. However, other results would have been much more difficult to prove. For instance, the argument on the smoothness of $\Lambda$ in \cite{minguzzi21} requires use of a connection $\nabla$ on the bundle $TH\to H$ which is induced from the spacetime connection thanks to the totally geodesic property of $H$ in the Lorentzian theory. However, there is no such natural connection in the Finslerian theory as recalled in the {\em Warning} above.

In other words, in some steps of the proofs one would need to introduce the spacetime $(U,g_n)$ as a tool to construct some useful quantities of non-physical interest. For instance, by  \cite[Lemma 3]{minguzzi21} we have
\begin{theorem}
Let  $H$ be a $C^2$ totally geodesic null hypersurface on a spacetime $(M,\mathscr{L})$. Then extending $n$ as shown above, for
vector fields $X,Y\in TH$
\begin{align}
R^{g_n}(X,Y)n&=\dd \omega(X,Y) n, \\
Ric^{g_n}(Y,n)&=\dd \omega(n,Y) .
\end{align}
Moreover, there is a bilinear form $\mu: H \to T^*H\otimes T^*H$ such that for every $X,Y\in TH$, $R^{g_n}(n,X)Y=\mu(X,Y) n$ and $\mu(X,Y)-\mu(Y,X)=-\dd \omega(X,Y)$.
\end{theorem}
We see that we have defined a 1-form $\mu$ that could not be defined in the Finslerian spacetime. This would be used in proofs and then removed from useful statements.
It is certainly more convenient to apply Theorem \ref{nner} directly.

Note also that  if $(M,\mathscr{L})$ satisfies $Ric_n\vert_{TH}=0$ by Thm.\ \ref{vie} $i_n\dd \omega=0$ and by the previous theorem  $Ric^{g_n}(n,\cdot)\vert_{TH}=0$ (this is condition ($\star \star$) in \cite{minguzzi21}).

\section{Translation of results} \label{trasl}

We list a number of corollaries which follow from  Lorentzian results obtained in previous work by the author in collaboration with Sebastian Gurriaran and Raymond Hounnonkpe \cite{minguzzi21,minguzzi24c,minguzzi25d}.

For shortness we introduce the following terminology
\begin{definition}
A {\em horizon} is a smooth embedded connected  totally geodesic null hypersurface.
\end{definition}
and the following property
\begin{itemize}
\item[$\star$] On the smooth spacetime $(M,\mathscr{L})$  let  $H$ be a horizon. We assume the `null convergence condition and Eq.\  (\ref{two})' or the weaker condition on the Ricci 1-form: $Ric_n\vert_{TH}=0$.
\end{itemize}

The following comes from \cite[Lemma 2]{minguzzi21}

\begin{lemma} \label{cpx}
Assume ($\star$).
Let $x\colon\mathbb{R}\to H$, $s \mapsto x(s)$,  be an integral curve of $n$, with $x(0)=p$.
The geodesic $\gamma$ starting from $x(\tau)$ with tangent $\dot \gamma=n$ has future affine length
\begin{equation} \label{bud}
\Lambda(x(\tau))=\int_\tau^\infty e^{\int_\tau^\rho \kappa(x(s))\dd s} \dd \rho.
\end{equation}
Thus $\Lambda$ is finite at one point of the integral curve iff it is finite everywhere over it, and in this case it satisfies the differential equation
\begin{equation} \label{biw}
1+\kappa \Lambda+\p_n \Lambda=0.
\end{equation}
In particular, for $\kappa$ constant, we have $\Lambda<\infty$ iff $\kappa<0$, in which case $\Lambda=
- \kappa^{-1}$.
\end{lemma}

The following comes from \cite[Lemma 6]{minguzzi21}

\begin{lemma} \label{loq}
Assume ($\star$). We have
$L_n\omega=\dd \kappa$ and hence $L_n \dd \omega=0$.
\end{lemma}

The following comes from \cite[Lemma 2 and Prop.\ 1]{minguzzi21}

\begin{proposition} \label{vos}
Assume ($\star$) and $H$ is compact. The validity of the property ``the integral $\int_{\gamma_p([0,s))} \omega$ over the future-directed generator  starting from $p$ converges to  $-\infty$'' (resp. $+\infty$, is upper bounded, is lower bounded) does not depend on the generator $\gamma_p$ considered. Moreover, if the horizon  admits one future incomplete generator, then $\int_{\gamma([0,\infty))} \omega=-\infty$ over every generator.
\end{proposition}




%

The following comes from \cite[Cor.\ 2 and 4]{minguzzi21}

\begin{corollary} \label{cbv}Assume ($\star$) and $H$ is compact.
Assume $H$ admits a future incomplete generator.
All geodesic generators are future incomplete. Moreover,  $\Lambda$ is smooth.
\end{corollary}
The last result expresses the {\em dichotomy} in the smooth category \cite{reiris21,minguzzi21} (see Isenberg and Moncrief \cite{moncrief08} for the original result  in the vacuum analytic setting). Since it holds also in the past version, the generators are either all complete or all incomplete (and in the latter case all future incomplete or all past incomplete).



The following comes from \cite[Def.\ 4]{minguzzi21} and corresponding equivalence proof.

\begin{theorem} \label{der}
 Assume ($\star$) and $H$ is compact.
The following properties are equivalent:
\begin{enumerate}
\item $H$ admits a future incomplete generator (and hence every generator is future incomplete),
\item The smooth vector field $n$ can be chosen such that $\kappa<0$ over $H$,
\item The smooth vector field $n$ can be chosen such that $\kappa=-1$ over $H$,
\item For some choice of smooth vector field $n$ (and hence for  every choice) there is a generator over which $\int_0^\infty \kappa(s)\dd s=-\infty$ (and hence the same is true for every generator).
\end{enumerate}
\end{theorem}

\begin{definition}
If they hold, we call
$H$ {\em future non-degenerate} (and analogously for the past). If they fail to hold in the future version and also in  the past version, then
$H$ is called {\em degenerate}.

\end{definition}




We arrive at a central result of this work which is the Finslerian generalization of   \cite[Thm.\ 4]{minguzzi25d}

\begin{corollary} \label{ocvt}
 Assume ($\star$) and $H$ is compact. Then there exists  a smooth future-directed lightlike tangent vector field $n$ of constant surface gravity: $\nabla_n n=\kappa n$,  $\kappa=cnst$. We have $\kappa=0$ iff the horizon is degenerate.
\end{corollary}

Thus under ($\star$) and compact  $H$ there are three cases: the generators are future incomplete and surface gravity can be normalized to $-1$ (non-degenerate case); the generators are future complete and surface gravity can be normalized to $0$ (degenerate case);  the generators are past incomplete and surface gravity can be normalized to $1$ (non-degenerate case).

We refer the reader to \cite{minguzzi24c} for an introduction to Riemannian flows. Essentially, on a null hypersurface we have  a Riemannian flow if there is a metric acting on the vector bundle $TH/n\to H$  preserved by the oriented 1-dimensional foliation  $\mathcal{F}$ of the generators.

The following follows from \cite[Prop.\ 1]{minguzzi24c}

\begin{proposition} \label{bort}
Assume ($\star$). The flow induced by the generators  is a Riemannian flow with a preassigned invariant metric $h$.
Conversely, every smooth null hypersurface $H$ that induces  a Riemannian flow structure $(H,h)$ is necessarily totally geodesic.
\end{proposition}

We recall that our argument goes as follows. Our horizon can be seen as a horizon in a Lorentzian manifold and since the previous proposition holds for such horizons and all terms involved in the statement are left invariant in the translation, the Finslerian version holds.

As observed by R. Hounnonkpe, the structural results for compact horizons based on the Riemannian flow property presented in \cite{minguzzi24c}  pass to the Finslerian case more directly because of characterization (iii) of Prop \ref{cmop}. He will pursue this line of reasoning in forthcoming work.




While this Riemannian flow observation  is sufficient to get some relevant results (e.g.\ Thms.\ \ref{mpty}-\ref{car} below), for others, such as the proof that   non-degenerate horizons admit non-vanishing constant surface gravity or Thm.\ \ref{mqp} and Cor.\ \ref{vmx} below, the full strength of our Theorem \ref{nndr} seems required.


The following is a generalization of a classical result \cite{petersen18b,minguzzi21,reiris21}. It can be obtained from the observation that, when $\kappa$ is a non-zero constant, thanks to Lemma \ref{loq} and Theorem \ref{der}, $L_n (g_n+\omega\otimes \omega)=0$, where the metric in parenthesis is positive definite

\begin{theorem} \label{mqp}
 Assume ($\star$) and $H$ is compact and non-degenerate. The Riemannian flow on $H$ is isometric.
\end{theorem}

The following comes from  \cite[Prop.\ 8]{minguzzi24c} which we obtained applying a classical result on foliations by Ghys

\begin{theorem} \label{mpty}
Assume ($\star$) and $H$ is compact and simply connected.  Then its Riemannian flow is isometric.
\end{theorem}

The Finslerian generalization continues as follows  \cite[Thm.\ 10]{minguzzi24c}

\begin{theorem}
Assume ($\star$) and $H$ is compact. The minimal volume in Gromov's sense of $H$ is zero, so $H$ does not admit a metric with negative sectional curvature.
\end{theorem}

We have also \cite[Thm.\ 13]{minguzzi24c}

\begin{theorem}
Assume ($\star$) and $H$ is compact.
The closures of
the leaves of $\mathcal{F}$ partition $H$ into embedded submanifolds that are diffeomorphic
to tori. The foliation induced on each torus is differentiably conjugate (without
parameter) to a linear flow on the torus with dense orbits.
\end{theorem}

In the following result $k+1$ is the dimension of the closures of the leaves $\bar L\simeq \mathbb{T}^{k+1}$, $L
\in \mathcal{F}$, of maximal dimension (they are tori by the previous theorem).
The main structure theorem in dimension 4 follows from \cite[Thm.\ 14]{minguzzi24c}

\begin{theorem} \label{car}
Assume ($\star$) and $H$ is compact and orientable, and that $M$ is $4$-dimensional. There are the following mutually excluding possibilities:
\begin{enumerate}
\item[(i)] All orbits are closed, $k=0$. $H$ is a Seifert fibration over a Satake manifold. The  fibers are the orbits of the flow.
\item[(ii)] The flow has precisely two closed orbits and every other orbit densely fills a two-torus, $k=1$. If the closed orbits are removed the manifold is diffeomorphic to a product $(0,1) \times  T^2$ where the flow in each fiber $\{s\} \times T^2$ is a linear flow with dense orbits. The manifold $H$ is obtained by gluing two solid tori by their boundaries, on each solid torus the foliation being given through a suspension by an irrational rotation of the disk. Moreover, there are two subcases:
\begin{itemize}
\item[a)] $H$ is diffeomorphic to the lens space\footnote{It comprises the case $L(1,0)=S^3$ while the case $L(0,1)=S^1 \times S^2$ is included in (ii-b). See \cite{saveliev12} for this gluing construction. Note that a $S^3$ topology can also appear in case (i) (e.g.\ Hopf bundle).} $L(p,q)$,  and the flow is conjugate to that of \cite[Example 1]{minguzzi24c}.
\item[b)] $H$ is diffeomorphic to $S^1\times S^2$ and the flow is conjugate to the flow given by the suspension of an irrational rotation of $S^2$ (with respect to, say, the $z$-axis in the standard isometric embedding in $\mathbb{R}^3$).
\end{itemize}
\item[(iii)] Every orbit densely fills a two-torus, so no orbit is closed, no orbit is dense,  $k=1$. $H$ is a $T^2$-bundle over $S^1$, each  torus fiber being the closure of the orbits in it, and there are two possibilities:
\begin{itemize}
\item[a)] $H$ is diffeomorphic to $T^3$ with a flow conjugate to a linear flow on the torus.
\item[b)] $H$ is diffeomorphic to the hyperbolic fibration $T^3_A$, ($
\textrm{tr} A > 2$) and the flow is conjugate to one of the flows of \cite[Example 2]{minguzzi24c}.
\end{itemize}
\item[(iv)] All orbits are dense, $k=2$. $H$ is diffeomorphic to $T^3$ with a flow conjugate to a linear flow on the torus.
\end{enumerate}
Moreover, except case (iii-b), the flow is isometric.
\end{theorem}

As a consequence  \cite[Cor.\ 15]{minguzzi24c} \cite{bustamante21}

\begin{corollary} \label{vmx}

Assume ($\star$) and $H$ is compact, non-degenerate and orientable, and that $M$ is $4$-dimensional. The Riemannian flow on $H$ has the  possible structures of Thm.\ \ref{car} save for case (iii-b) that does not apply, that is
\begin{itemize}
\item[i)]  if all the generators are closed  then $H$ is a Seifert manifold,
\item[ii)] if only two generators are closed and every other generator densely fills a two-torus then $H$ is a lens space,
\item[iii)] if every generator densely fills a two-torus, then $H$ is a $\mathbb{T}^2$-bundle over $S^1$ (and diffeomorphic to $\mathbb{T}^3$),
\item[iv)] If every generator densely fills the horizon, then $H$ is a three-torus $\mathbb{T}^3$.
\end{itemize}
\end{corollary}

For completeness we mention the classification for spacetime dimension 2 and 3 which comes from \cite{minguzzi25d}.

\begin{theorem} \label{carf}
Assume ($\star$) and $H$ is compact and orientable. In spacetime dimension 2 the manifold $H$ is diffeomorphic to $S^1$ while in spacetime dimension 3 the manifold $H$ is diffeomorphic to a torus with the flow being conjugate to (a) a constant flow  along the fibers of the second projection  $S^1\times S^1\to S^1$, (b) a linear flow over $T^2$ with dense orbits.
In all cases the flow is isometric.
\end{theorem}

The following result is interesting already in the Lorentzian case (as it is more general than \cite[Lemma 2.8]{petersen18} obtained for the Lorentzian non-degenerate case). Observe that, in view of Theorem \ref{car}, in 4 spacetime dimensions  only the special non-isometric (and hence degenerate) case (iii-b) is left out.
\begin{theorem} \label{cnpp}
Suppose that on the compact horizon $H$  there exist a future-directed lightlike tangent field $n$ and a Riemannian metric $h$ such that $L_n h=0$ (isometric flow). Then $n$ can be extended in a neighborhood of $H$ in such a way that $L_n g_n\vert_H=0$. In particular, assume  ($\star$)  and suppose $H$ is non-degenerate,  if $n$ is the constant surface gravity field then we can extend it so that  $L_n g_n\vert_H=0$.
\end{theorem}

The condition  $L_n g_n=0$ is equivalent to $n^*(L_{n^c} g)=0$ where $n^c$ is the {\em complete lift}. For   Knebelman's definition of Lie derivative and  Killing vector field for pseudo-Finsler spaces see \cite[Eq.\ (10.8)]{knebelman29} \cite{rund59,caponio18}. As a consequence, there seems to be little hope of proving that $n$ is Killing in a neighborhood of $H$, as would posit a Finslerian version of the Isenberg-Moncrief conjecture \cite{moncrief83}, unless we are in the Lorentz rather than Lorentz-Finsler case. The most natural expectation is that there must be a vector field $n$ in a neighborhood of $H$ such that $L_ng_n=n^*(L_{n^c}g)=0$.

In the isometric degenerate case we are not claiming that $n$ such that $L_n g_n\vert_H=0$ coincides with the geodesic lightlike tangent field (which exists by Cor.\ \ref{ocvt}).
\begin{proof}
For the first statement, without los of generality we can assume $h(n,n)=1$, otherwise replace $h$ with $h'=h/h(n,n)$ which is still such that $L_n h'=0$. Observe that the 1-form $\sigma:=h(n,\cdot)$ is such that $\sigma(n)=1$ and $L_n \sigma=0$.
 Let $N$ be the future-directed  $g_n$-lightlike field on $H$ defined by $g_n(N,\cdot)\vert_{TH}=-\sigma$, so that $g_n(N,n)=-1$. Then extend $N$ in any way, for instance geodesically, and extend $n$ so that $L_Nn=0$.

 We already know that for every $X,Y\in TH$, $L_n g_n(X,Y)=0$, see Prop.\ \ref{cmop}. It remains to prove $L_ng_n(N,X)=0$ for $X\in TH$ and $L_ng_n(N,N)=0$. For the former equation, we have on $H$
\begin{align*}
L_ng_n(N,X)&=n(g_n(N,X)-g_n(L_nN,X)-g_n(N,L_nX)\\
&=n(g_n(N,X))-g_n(N,L_nX)\\
&=-[n(\sigma(X))-\sigma(L_nX)]=-(L_n\sigma)(X)=0.
\end{align*}
For the second equation, we have on $H$
\[
L_ng_n(N,N)=n(g_n(N,N))-2g_n(L_nN,N)=0.
\]
The last statement follows taking  $c:=\kappa=\omega(n)$ and defining $h=g\vert_{TH\times TH}+\frac{1}{c^2} \omega \otimes \omega$, see also Theorem \ref{mqp}.
\end{proof}

\begin{remark}
Results  on the existence of Killing vectors for the quotient metric $h$ (or the induced metric $g_n\vert_{TH\times TH}$) obtained in \cite[Sec.\ 2.4]{minguzzi24c} can also be generalized to the Finslerian case but will not be recalled here.
\end{remark}

\section{Physical considerations on the dominant energy condition} \label{kxrt}

In the previous sections, we have shown that the equation \(\chi_\alpha = 0\) naturally leads, via the null convergence condition, to the constancy of surface gravity for compact horizons, as well as to several other topological results concerning the structure of these hypersurfaces.

The situation closely resembles that of Lorentzian geometry. However, it remains possible that the condition \(\chi_\alpha = 0\) may not even be required. Recall that the sufficiency of the null convergence condition for establishing the constancy of surface gravity was only recently realized \cite{minguzzi25d}. Prior to this, it was widely believed that even in the case of Killing horizons, this result should follow from the stronger dominant energy condition.

In this section, we aim to explore the following question: what can be inferred about the gravitational field equations if the constancy of surface gravity follows from the dominant energy condition instead? The fact that, in Lorentzian geometry, this result follows merely from the null convergence condition could be viewed as an algebraic accident. By imposing a stronger energy condition, it may become unnecessary to require \(\chi_\alpha = 0\) in the presence of matter.

As it turns out, this idea proves fruitful and leads to a constrained family of field equations (Thm.\ \ref{pcfq}). To begin with, let us introduce a velocity-dependent energy-momentum \(1\)-form \((T_v)_\alpha\), which in general relativity would correspond to \(T_{\alpha \beta}(x) v^\beta\).
%
%
%
%
%
%
%
%

The 1-form $T$ is a source term that we should connect with geometric quantities. We define
\begin{definition}
The {\em dominant energy condition} is satisfied if $-T_v^\sharp$ is a future-directed causal vector or the zero vector for every future-directed causal vector $v$. The {\em weak energy condition} is satisfied if $T_v(v)=g_v(v, T_v^\sharp)\ge 0$ for every  future-directed causal vector $v$.  The {\em null energy condition} is satisfied if $T_v(v)=g_v(v, T_v^\sharp)\ge 0$ for every  future-directed lightlike vector $v$.
\end{definition}
Of course, with $-T_v^\sharp$ we mean the vector of components $- T_v^\alpha$ where the index is raised with the inverse of $g_v$.
Note that we do not ask $T^\alpha$ to be a past-directed causal vector since the Finsler Lagrangian is not assumed to be reversible. As we want all the physics to come from the geometry of the future causal cone we are lead to  the above definition.

Observe that vacuum is defined by $T_v=0$ and in this case all previous energy conditions are satisfied.

More generally, by the  reverse Cauchy-Schwarz inequality \cite{minguzzi13c} the dominant energy condition implies the weak energy condition.  Clearly, the weak energy condition implies the null energy condition.


In general relativity, for a null vector $v$ we have the identity $8\pi (T_v)_\alpha v^\alpha=Ric(v)$ which seems reasonable to expect also in a Finslerian framework.

\begin{proposition} \label{pvge}
Assume the dominant energy condition holds and that $8\pi(T_v)_\alpha v^\alpha=Ric(v)$ for every future-directed lightlike vector $v$.
On a totally geodesic hypersurface $H$ with future-directed lightlike tangent field $n$, $(T_n)_\alpha \propto n_\alpha$.
\end{proposition}
As it is customary, we have denoted $v_\alpha:=(g_v)_{\alpha \beta} v^\beta$.

\begin{proof}
Indeed, on $H$ we have $Ric(n)=0$ by Prop.\ \ref{cmop}, hence $(g_n)_{\alpha \beta} (-(T_n)^\alpha) n^\beta=0$. Observe that both $-(T_n)^\alpha$ and $n^\alpha$ are future-directed causal. However, by the equality case of the reverse Cauchy-Schwarz inequality \cite[Thm.\ 3]{minguzzi13c}  $-(T_n)^\alpha\propto n^\alpha\ne 0$.
\end{proof}

In order to accomplish the constancy of surface gravity all is needed is the equality $(Ric_n)_\alpha \propto n_\alpha$ at $H$, which is naturally accomplished if $Ric_n$ is proportional to a linear combination of $T_n$ and $g_n(n, \cdot)$. Thanks to Prop.\ \ref{cnph} a more general result holds

\begin{theorem} \label{pcfq}
Assume the dominant energy condition.
If the equation\footnote{We could add a further term $\mathscr{L} B_v$ on the left-hand side as it  vanishes over lightlike vectors, but it does not reduce to existing terms in general relativity, unless $B_v \propto g_v (v,\cdot)$ in which case it can be absorbed in that already present. }
\begin{equation} \label{npuy}
R^\mu_{\mu \alpha}- a(v) v_\alpha + b(v) R^\mu{}_\nu C^\nu_{\mu \alpha}+c(v) R^\mu{}_\nu L^\nu_{\mu \alpha}+d(v) R^\mu{}_\alpha I_\mu+e(v) R^\mu{}_\alpha J_\mu= 8\pi T_\alpha
\end{equation}
 holds for suitable functions $a,b,c,d,e$, and for every  future-directed lightlike vector $v$, then the previous results of this work, including the constancy of surface gravity for compact totally geodesic hypersurfaces, hold.
\end{theorem}
Note that we are not requiring $\chi_\alpha=0$ and, as a consequence, we are placing much more severe constraints on the gravitational equations, as we assume that, as in general relativity, the recovery of constancy of surface gravity should not depend on a vacuum assumption.

\begin{proof}
Contracting the equation with a lightlike vector $v^\alpha$ gives at support vector $v$, $Ric(v)=8\pi(T_v)_\alpha v^\alpha$. By Prop.\ \ref{pvge}, on $H$ choosing $v=n$,  $(T_n)_\alpha \propto n_\alpha$. Using Prop.\ \ref{cnph}  we get $(Ric_n)_\alpha \propto n_\alpha$.
\end{proof}

\begin{proposition}
Assume Eq.\ (\ref{npuy}). The null energy condition and the null convergence condition are equivalent. 
\end{proposition}
In vacuum ($T=0$) Eq. (\ref{npuy}) implies $Ric(v)=2 a(v) \mathscr{L}$. For some choice of coefficients it can also imply $Ric(v)=0$, see Example \ref{cbtr}.
\begin{proof}
Contract (\ref{npuy}) with $v$ lightlike, then  $Ric(v)=8\pi T_v(v)$.
We conclude that $Ric(v) \ge 0$ iff $T_v(v)\ge 0$, where the inequality on the left expresses the null convergence condition while that on the right expresses the null energy condition.
\end{proof}

We wish to stress the subtle difference between the null convergence condition, which is a condition on the geometry, and the null energy condition, which is a condition on the nature of the source. In general relativity,  the two are known to be equivalent and so these terms are often used interchangeably.  However, it is really dependent on the field equations whether they are equivalent or not. Fortunately, in a gravitational theory based on (\ref{npuy}) they are.

The found family is significant as it contains the equation we proposed in \cite[Sec.\ 5.4.2]{minguzzi14c}
provided $L=0$. Let us express this with a theorem
\begin{theorem} \label{vxe}
On a Landsberg pseudo-Finsler space the tensor
\begin{equation} \label{npx}
E^\beta{}_\alpha:=\frac{1}{2} (R^{HH} {}^{\mu \beta}{}_{\mu \alpha} +R^{HH} {}^{\beta \mu}{}_{\alpha \mu}-\delta^\beta_\alpha R^{HH}  {}^{\mu \nu}{}_{\mu \nu})
\end{equation}
satisfies $\nabla^{H}_\beta E^\beta{}_\alpha=0$, so the current $8\pi T^\beta:= E^\beta{}_\alpha v^\alpha$ satisfies  $\nabla^{H}_\alpha T^\alpha=0$.
Additionally,
\begin{equation} \label{cmxx}
8\pi T_\beta=R^\mu{}_{\mu \beta} +R_{\mu  \nu} C^{\mu\nu}_\beta-R^\mu_\beta I_\mu - \frac{1}{2}v_\beta R^{HH}  {}^{\mu \nu}{}_{\mu \nu} ,
\end{equation}
which has the  form (\ref{npuy}) (for $a=\frac{1}{2} R^{HH}  {}^{\mu \nu}{}_{\mu \nu}$, $b=-d=1$, $c=e=0$).
\end{theorem}

Note that in the theorem we do not distinguish between horizontal Berwald or Chern-Rund covariant derivatives as $L=0$. Note also that  $R^{HH}  {}^{\mu \nu}{}_{\mu \nu}=g^{\mu \nu}\frac{\p}{\p v^\mu} R^\sigma{}_{\sigma \nu}$ and in the Lorentzian case the current reduces itself to that of General Relativity.


\begin{proof}
By the contracted HHH-second Bianchi identity \cite[Eq.\ (112)]{minguzzi14c}
\begin{equation} \label{kkg}
2\nabla^H_\beta E^\beta{}_\alpha={G}^{\mu \nu}{}_{ \delta \alpha} R^\delta_{\mu
\nu}+({G}^{\delta \nu}{}_{ \sigma \delta}
-{G}^{ \nu \delta}{}_{ \sigma \delta} )R^\sigma_{ \nu \alpha}=0.
\end{equation}
The last equality uses only the total symmetry of the Berwald curvature ${G}_{\alpha \beta \gamma \delta}$, a property which is equivalent to the Landsberg condition \cite[Prop.\ 12.5]{vattamany04}\cite[Prop.\ 2.4]{crampin11} \cite[Eq.\ (58)]{minguzzi14c}.
 Using the HH-symmetry  \cite[Eq.\ (87)]{minguzzi14c},
\[
R^{HH}{}_{\beta \mu \alpha \nu}-R^{HH}{}_{\mu \beta  \nu \alpha}=-2R^
\sigma_{\alpha \nu} C_{\beta \mu \sigma},
\]
contracting $\mu$ with $\nu$, we get  $R^{HH}{}_\beta{}^{\mu}{}_{\alpha
\mu}-R^{HH}{}^\mu{}_{\beta  \mu \alpha}=-2R^
\mu_{\alpha \nu} C^\nu_{\beta  \mu}$.  Moreover \cite[Eq.\ (91)]{minguzzi14c}, $R^{HH}{}^\mu{}_{\beta \mu \alpha}-R^{HH}{}^\mu{}_{\alpha \mu \beta}=R^\mu_{\alpha \beta} I_\mu$, so using it
\[
R^{HH}{}^\mu{}_{\alpha
\mu \beta}-R^{HH}{}_\beta{}^{\mu}{}_{\alpha \mu}=2R^
\mu_{\alpha \nu} C^\nu_{\beta \mu}+R^\mu_{\beta\alpha}I_\mu.
\]
Using it to replace the  second term of (\ref{npx}) we get Eq.\ (\ref{cmxx}) (note that, again by \cite[Eq.\ (91)]{minguzzi14c}, $R^{HH} {}^{\mu}{}_{\beta\mu \alpha} v^\alpha=R^{HH} {}^{\mu}{}_{\alpha \mu \beta} v^\alpha -R^\mu_\beta I_\mu =R^\mu_{\mu \beta}-R^\mu_\beta I_\mu $).
\end{proof}

\begin{remark}
An analogous study of the HHH-Bianchi identity for the Cartan connection without imposing the Landsberg condition leads to
\begin{equation}
8\pi \nabla^{HC}_\mu T^\mu= R^{VH}_{\textrm{Car}}{}^{\nu}{}_{\alpha \beta\nu} R^{\alpha \beta}.
\end{equation}
with the more general vector field
\begin{align}
8\pi T_\mu&=\frac{1}{2} \big(R^{HH}_{\textrm{Car}}{}^{\sigma}{}_{\mu\sigma\alpha}+R_{\textrm{Car}}^{HH}{}_\mu{}^{\sigma}{}_{\alpha\sigma}- g_{\mu\alpha} R_{\textrm{Car}}^{HH}{}^{\gamma\delta}{}_{\gamma\delta}\big) v^\alpha\\
&= \big(R^{HH}_{\textrm{Car}}{}^{\sigma}{}_{\mu\sigma\alpha}-\frac{1}{2}  g_{\mu\alpha} R^{HH}_{\textrm{Car}}{}^{\gamma\delta}{}_{\gamma\delta}\big) v^\alpha \label{ckt1}\\
&=R^\sigma{}_{\sigma\mu}+R^\delta_\sigma C^\sigma_{\delta\mu}- R^\sigma_\mu I_\sigma -\frac{1}{2}  v_\mu R^{HH}_{\textrm{Car}}{}^{\gamma\delta}{}_{\gamma\delta}, \label{ckt2}
\end{align}
which is again of the form (\ref{npuy}). However, it  has no vanishing Cartan horizontal divergence, in general. Since $R^{HH}_{\textrm{Ber}}{}^{\alpha\beta}{}_{\alpha \beta}=R^{HH}_{\textrm{Car}}{}^{\alpha\beta}{}_{\alpha\beta}+ J_\alpha J^\alpha-L_{\alpha\beta\gamma} L^{\alpha\beta\gamma}$, for $L=0$ (which is equivalent to $R^{VH}_{\textrm{Car}}=0$ \cite[Thm.\ 25.3]{matsumoto86}) it is horizontally conserved and  becomes coincident with (\ref{cmxx}). We note that using $R^\sigma{}_{\sigma\beta}=-2 \chi_\beta-\frac{\p}{\p v^\sigma}R^\sigma_\beta$ it can be rewritten in the more elegant form
\begin{equation}
8\pi T_\beta=-2\chi_\beta-\nabla^{VC}_\sigma R^\sigma_\beta-\frac{1}{2}  v_\beta R^{HH}_{\textrm{Car}}{}^{\gamma\delta}{}_{\gamma\delta},
\end{equation}
where $\nabla^{VC}$ is the vertical Cartan covariant derivative.
\end{remark}

\begin{example}
A result by Rund \cite{rund62} (see also \cite{matsumoto86,takano88}), later slightly generalized by Shibata and other  authors  \cite{shibata78,ishikawa80b,ikeda81}, states that for constant curvature spaces \cite[Sec.\ 26]{matsumoto86}, i.e.\ having curvature endomorphism of the form $R_{\alpha\beta}=2K\mathscr{L} h_{\alpha\beta}$, where $K$ is a constant, a certain Einstein's type tensor is Cartan horizontally conserved. Its expression leads to the  Cartan horizontally  conserved current
\begin{equation}
8\pi T'{}^\sigma=R^{HH}_{\textrm{Car}}{}^{\mu\sigma}{}_{\mu\alpha} v^\alpha -\frac{1}{2} v^\sigma  \big( R^{HH}_{\textrm{Car}}{}^{\mu \nu}{}_{\mu\nu}+   K 2 \mathscr{L}   R^{VV}_{\textrm{Car}}{}^{\mu\nu}{}_{\mu\nu}\big),
\end{equation}
which, due to the equality between (\ref{ckt1}) and (\ref{ckt2}),  is also of the form (\ref{npuy}). The imposed condition is particularly strong but still the result might be of interest in some applications.
\end{example}

\begin{example} \label{cbtr}
Let us consider a simplifying case. Let us determine $a(v)$ in vacuum under the assumption $b=c=d=e=0$ and that Eq.\ (\ref{npuy}) holds for every causal $v$. Function $a(v)$ can only depend on the  geometry as there is no  source content. The  vacuum equation is
\[
R^\mu{}_{\mu \alpha}=a(v) v_\alpha .
\]
Observe that $a(v)$ is positive homogeneous of  degree zero. Differentiating
\[
\frac{\p}{\p v^\beta} R^\mu{}_{\mu \alpha}= v_\alpha \frac{\p a }{\p v^\beta} + a g_{\alpha \beta},
\]
thus contracting with $g^{\alpha \beta}$ (the spacetime dimension is $n+1$)
\[
g^{\alpha \beta}\frac{\p}{\p v^\beta} R^\mu{}_{\mu \alpha}=v^\beta\frac{\p a }{\p v^\beta} + a (n+1)=a(n+1)\quad \Rightarrow \quad a(v)=\frac{1}{n+1} g^{\alpha \beta}\frac{\p}{\p v^\beta} R^\mu{}_{\mu \alpha}
\]
If in the original vacuum equation $a(v)$ had precisely this expression, this equation is adding nothing. It it particularly convenient to explore the case in which $a(v)=k g^{\alpha \beta}\frac{\p}{\p v^\beta} R^\mu{}_{\mu \alpha}$ with $k\ne \frac{1}{n+1}$ for  in this case the previous equation establishes that $g^{\alpha \beta}\frac{\p}{\p v^\beta} R^\mu{}_{\mu \alpha}=0$ and so the vacuum equation leads to $R^\mu{}_{\mu \alpha}=0$ and so implies also $Ric(v)=0$, which is very desirable. In general relativity $a(v)$ would be half the scalar curvature, thus we are led to $k=1/2$ (in spacetime dimension different from 2), and to the Finslerian gravitational equation
\begin{equation} \label{cmbt}
R^\mu{}_{\mu \alpha}-\frac{1}{2}  [g^{\mu \nu}\frac{\p}{\p v^\mu} R^\sigma{}_{\sigma \nu}] v_\alpha = 8\pi T_\alpha.
\end{equation}
\end{example}

It is important that our arguments based on energy conditions brought us to vectorial equations. Some vectorial or tensor equation are equivalent to scalar equations (we shall see an example below), but in general in Finsler gravity it was not even clear if the gravitational equations had to be tensorial (most proposals particularly in the seventies and eighties), vectorial (we suggested to use a velocity-dependent energy-momentum in \cite{minguzzi14c,minguzzi15d}) or scalar (notably \cite{rutz93,pfeifer12}).

It is appealing that the arguments leading to the above equations only involve the non-linear curvature as this sidesteps all indeterminacies that would come from algebraic analogies, as there is no need to choose a linear Finsler connection  to express the equations. We recall that via algebraic analogies basically all options have been explored, for instance Berwald \cite{li10b}, Chern-Rund \cite{horvath50}, and Cartan \cite{takano74,ikeda79,ishikawa80,ishikawa81,takano88}.

Those  using just the non-linear curvature  as this work, can be classified in all categories because all notable linear Finsler connections induce the same non-linear connection.
Of this type are  the equations one obtains from the action integral over the indicatrix $\int_I Ric(v)  \dd \nu$ where $\nu$ is a suitable  measure \cite{chen08b,pfeifer12,hohmann19,javaloyes21}.
Of this type are also those approaches that assume a Landsberg condition \cite{minguzzi14c} (making the Chern-Rund and Berwald options equivalent) or those that build the Einstein tensor via an algebraic analogy by making  use of the Akbar-Zadeh Ricci tensor \cite{li14,manjunatha23} $R^{AZ}_{\alpha \beta}=\frac{1}{2}\p^2  Ric(v)/\p v^\alpha \p v^\beta$ and the associated scalar curvature. The latter approach has been pursed in vacuum where the tensorial analogy is straightforward (it is a type of Einstein condition  considered by Akbar-Zadeh in positive signature \cite{akbarzadeh95,bao95})
\begin{equation} \label{cpkt}
\frac{\p^2  Ric(v)}{\p v^\alpha \p v^\beta}-\frac{1}{2} [g^{\mu \nu} \frac{\p^2  Ric(v)}{\p v^\mu \p v^\nu}] g_{\alpha \beta}=0 .
\end{equation}
However, by contracting with $g^{\alpha \beta}$ and using positive homogeneity one immediately finds that (in spacetime dimension different from 2) this equation is really equivalent to the simpler Ricci flat condition
\[
Ric(v)=0.
\]


Although the equations (\ref{cmxx}) or  (\ref{cmbt})  do not imply $\chi=0$, in vacuum for (\ref{cmbt}) this happens, as we have a result which clarifies the connection with $Ric(v)=0$.

\begin{theorem}
Suppose $n\ne 1$.
In vacuum ($T=0$) equation (\ref{cmbt}) is equivalent to $Ric_v=0$ which is equivalent to `$Ric(v)=0$ and $\chi=0$'.
\end{theorem}
Thus in the theory expressed by  Eq.\ (\ref{cmbt}) vacuum is expressed by the vanishing of the Ricci 1-form rather than the Ricci scalar.
\begin{proof}
Assume (\ref{cmbt}) in the vacuum case. By the same argument given above when discussing the form of $a(v)$ we have  $g^{\mu \nu}\frac{\p}{\p v^\mu} R^\sigma{}_{\sigma \nu}=0$ and so $R^\mu{}_{\mu \alpha}=0$. The converse is clear so this proves the first equivalence. Now, if $R^\mu{}_{\mu \alpha}=0$ as we have $Ric(v)=R^\mu{}_{\mu \alpha} v^\alpha$ and  Eq.\ (\ref{cnng}) $\chi_\alpha=\frac{\p}{\p v^{\alpha}} Ric(v) -2R^\gamma{}_{\gamma \alpha}$ the direction to the right of the second equivalence follows. For the converse, if $\chi_\alpha=0$, we have again by Eq.\ (\ref{cnng})    $\frac{\p}{\p v^{\alpha}} Ric(v)=2R^\gamma{}_{\gamma \alpha}$. As noted above  $Ric(v)=0$ is equivalent to  (\ref{cpkt}), where  substituting the just found expression for $\frac{\p}{\p v^{\alpha}} Ric(v)$  and multiplying by $v^\beta$ we get (\ref{cmbt}).
\end{proof}


Summarizing, there are two conclusions that can be drawn depending on the true physical origin of constant surface gravity in compact horizons
\begin{itemize}
\item[(i)]  If the constancy of surface gravity  follows from the {\bf null convergence condition}, then, we argued, the equation $\chi=0$ must hold, even not in vacuum, and should hence be imposed along with $Ric(v)=0$ in vacuum. Whether the null convergence condition is equivalent to the null energy condition depends on the other gravitational field equations.
\item[(ii)]  If the constancy of surface gravity  follows from the {\bf dominant energy condition}, then, we argued, there is no more reason to expect the equation $\chi=0$ in the non-vacuum case. However, the Finslerian gravitational equations are severely constrained as they need to compensate for the missing condition $\chi=0$ in presence of matter. They  include a family of possibilities (Eq.\ (\ref{npuy})) where $T^\sharp_v$ is the  energy-momentum. Its most notable cases are Eq.\ (\ref{cmxx}) and Eq.\ (\ref{cmbt}). Of course, if any of them  is  physically correct then it is expected to hold also when the dominant energy condition is violated. In vacuum  the latter equation is equivalent to the vanishing of the Ricci 1-form or to the pair given by  `$Ric(v)=0$ and $\chi=0$'.
\end{itemize}

\section{Conclusions}

In this work, we introduced the notions of  Ricci 1-form and totally geodesic null hypersurface for Finsler spacetimes. We proved that the equation
\[
Ric_n\vert_{TH}=0,
\]
 on a totally geodesic hypersurface is equivalent to $i_n\dd \omega=0$. Consequently, several techniques developed in Lorentzian geometry, including the ribbon argument, are expected to be translatable into a Finslerian framework. Rather than reworking all proofs, some of which would require rather non-trivial modifications, we devised a simple trick that allowed us to embed the hypersurface into a Lorentzian manifold, making the translation immediate. We presented several results of this kind, with a special focus on those pertaining to the problem of establishing the constancy of surface gravity in compact totally geodesic hypersurfaces. These results are of importance due to the interpretation of surface gravity as temperature, its constancy expressing the zeroth law of thermodynamics.

Having established this generalization, we then regarded the problem of deriving the equation $Ric_n\vert_{TH}=0$ as a novel physical constraint on potential Finslerian gravitational theories. We showed that this equation can be derived in two distinct ways by assuming that it should hold, as in general relativity, even in the presence of matter.

In the first case, we derived it from the null energy condition, while in the second, we derived it from the dominant energy condition. The former approach suggests that the equation $\chi=0$ must hold, even in the presence of matter. The latter approach suggests some more constrained possibilities expressed by Eq.\ (\ref{npuy}). The novel Eq.\ (\ref{cmbt}) is interesting, as it reduces to $\chi=0$ and $Ric=0$ in vacuum.  The option (\ref{cmxx}) is better behaved with regard to conservation laws (under $L=0$),  and still implies the constancy of  surface gravity.

In conclusion, this work has two facets. It is a mathematical work deriving many new detailed results on totally geodesic null hypersurfaces in Finsler spacetimes, and it is also a physical work supporting certain gravitational field equations in novel ways. It should be noted that physical support for gravitational equations can only come in the following ways: (a) finding special solutions and comparing them with experiment, and (b) establishing connections with very desirable physical conclusions.

Strategy (a) has been followed in several works, but so far it has not been able to direct attention to a specific set of equations. Strategy (b) has been followed mostly with regard to the criteria of (i) the validity of conservation laws deduced from the equations, and (ii) the variational derivation of the equations. Passing criterion (i) has proved particularly difficult, as no clear-cut conservation law has been obtained. As for (ii), variants have been obtained that do not reduce straightforwardly to general relativity and have yet to find experimental confirmation. Here, we have introduced and followed a novel direction (iii): establishing the constancy of surface gravity. In it we found  considerable constraints and support for specific equations which so far were not physically motivated. There is hope that further investigation will help shed light on the conservation laws induced by the proposed equations.

\section*{\normalsize Data availability statement}
No new data were created or analysed in this study.

\section*{\normalsize Conflict of Interest statement}
The author of this publication declares no conflict of interest.


\begin{thebibliography}{10}

\bibitem{aazami14}
A.~B. Aazami and M.~A. Javaloyes.
\newblock {P}enrose's singularity theorem in a {F}insler spacetime.
\newblock {\em Class. Quantum Grav.}, 33:025003, 2016.

\bibitem{akbarzadeh95}
H.~Akbar-Zadeh.
\newblock Generalized {E}instein manifolds.
\newblock {\em J. Geom. Phys.}, 17:342--380, 1995.

\bibitem{antonelli93}
P.~L. Antonelli, R.~S. Ingarden, and M.~Matsumoto.
\newblock {\em The Theory of Sprays and {F}insler Spaces with Applications in
  Physics and Biology}.
\newblock Springer Science+Business Media, Dordrecht, 1993.

\bibitem{bao95}
D.~Bao.
\newblock Review of {A}kbar-{Z}adeh's {Generalized Einstein manifolds}.
\newblock MR1365208, 1995.

\bibitem{bao00}
D.~Bao, S.-S. Chern, and Z.~Shen.
\newblock {\em An Introduction to {R}iemann-{F}insler Geometry}.
\newblock {Springer-Verlag}, New York, 2000.

\bibitem{beem70}
J.~K. Beem.
\newblock Indefinite {F}insler spaces and timelike spaces.
\newblock {\em Can. J. Math.}, 22:1035--1039, 1970.

\bibitem{bejancu14}
A.~Bejancu and H.~R. Farran.
\newblock Theory of {F}insler submanifolds via {B}erwald connection.
\newblock {\em International Electronic Journal of Geometry}, 7:108--125, 2014.

\bibitem{berk09}
G.~Berk.
\newblock Minimality of totally geodesic submanifolds in {F}insler geometry.
\newblock {\em Math. Ann.}, 343:955--973, 2009.

\bibitem{bustamante21}
I.~Bustamante and M.~Reiris.
\newblock A classification theorem for compact {C}auchy horizons in vacuum
  spacetimes.
\newblock {\em Gen. Relativ. Gravit.}, 53:36, 2021.

\bibitem{caponio20}
E.~Caponio and A.~Masiello.
\newblock On the analyticity of static solutions of a field equation in
  {F}insler gravity.
\newblock {\em Universe 2020, 6, 59}, 6:59, 2020.

\bibitem{caponio18}
E.~Caponio and G.~Stancarone.
\newblock On {F}insler spacetimes with a timelike {K}illing vector field.
\newblock {\em Class. Quantum Grav.}, 35:085007, 2018.

\bibitem{chen08b}
B.~Chen and Y.B. Shen.
\newblock On a class of critical {R}iemann-{F}insler metrics.
\newblock {\em Publ. Math. Debrecen}, 72:451--468, 2008.

\bibitem{chrusciel20}
P.~T. Chru{\'s}ciel.
\newblock {\em Geometry of black holes}.
\newblock Oxford {U}niversity {P}ress, Oxford, 2020.

\bibitem{crampin11}
M.~Crampin.
\newblock On {L}andsberg spaces and the {L}andsberg-{B}erwald problem.
\newblock {\em Houston J. Math.}, 37:1103--1124, 2011.

\bibitem{dahl06}
M.~Dahl.
\newblock A brief introduction to {F}insler geometry.
\newblock Based on licentiate thesis, `Propagation of {G}aussian beams using
  {R}iemann-{F}insler geometry', Helsinki University of technology, 2006.

\bibitem{friedrich99}
H.~Friedrich, I.~R{\'a}cz, and R.~M. Wald.
\newblock On the rigidity theorem for spacetimes with a stationary event
  horizon or a compact {C}auchy horizon.
\newblock {\em Comm. Math. Phys.}, 204(3):691--707, 1999.

\bibitem{fuster16}
A.~Fuster and C.~Pabst.
\newblock Finsler pp-waves.
\newblock {\em Phys. Rev. D}, 94:104072, 2016.

\bibitem{fuster18}
A.~Fuster, C.~Pabst, and C.~Pfeifer.
\newblock Berwald spacetimes and very special relativity.
\newblock {\em Phys. Rev. D}, 98:084062, 2018.

\bibitem{galloway00}
G.~J. Galloway.
\newblock Maximum principles for null hypersurfaces and null splitting
  theorems.
\newblock {\em Ann. {H}enri {P}oincar\'e}, 1:543--567, 2000.

\bibitem{minguzzi21}
S.~Gurriaran and E.~Minguzzi.
\newblock Surface gravity of compact non-degenerate horizons under the dominant
  energy condition.
\newblock {\em Commun. Math. Phys.}, 395:679--713, 2022.
\newblock {arXiv:}2108.04056.

\bibitem{hohmann19}
M.~Hohmann, C.~Pfeifer, and N.~Voicu.
\newblock Finsler gravity action from variational completion.
\newblock {\em Phys. Rev. D}, 100:064035, 2019.

\bibitem{horvath50}
J.~I. Horv\'ath.
\newblock A geometrical model for the unified theory of physical fields.
\newblock {\em Phys. Rev.}, 80:901, 1950.

\bibitem{minguzzi24c}
R.~A. Hounnonkpe and E.~Minguzzi.
\newblock Horizon classification via {R}iemannian flows.
\newblock {\em Ann. Henri Poincar\'e}, 2024.
\newblock {arXiv:}2410.07231, DOI:10.1007/s00023-025-01570-2.

\bibitem{minguzzi24}
R.~A. Hounnonkpe and E.~Minguzzi.
\newblock Regularity and temperature of stationary black hole event horizons.
\newblock {\em Commun. Math. Phys.}, 406:228, 2025.
\newblock {arXiv:}2408.01252.

\bibitem{minguzzi25d}
R.~A. Hounnonkpe and E.~Minguzzi.
\newblock Compact {C}auchy horizons admit constant surface gravity.
\newblock {\em Commun. Math. Phys.}, 2026.
\newblock {arXiv:}2506.20004.

\bibitem{ikeda79}
S.~Ikeda.
\newblock On the theory of gravitational field in {F}insler spaces.
\newblock {\em Lett. Nuovo Cimento}, 26(9):277--281, 1979.

\bibitem{ikeda81}
S.~Ikeda.
\newblock On the conservation laws in the theory of fields in {F}insler spaces.
\newblock {\em J. Math. Phys.}, 22(6):1211--1214, 1981.

\bibitem{ingarden93}
R.~S. Ingarden and M.~Matsumoto.
\newblock On the 1953 {B}arthel connection of a {F}insler space and its
  mathematical and physical interpretation.
\newblock {\em Rep. Math. Phys.}, 32:35--48, 1993.

\bibitem{ishikawa80}
H.~Ishikawa.
\newblock Einstein equation in lifted {F}insler spaces.
\newblock {\em Il Nuovo Cimento}, 56:252--262, 1980.

\bibitem{ishikawa80b}
H.~Ishikawa.
\newblock Einstein tensor in scalar curvature {F}insler spaces.
\newblock {\em Ann. Physik (7)}, 37(2):151--154, 1980.

\bibitem{ishikawa81}
H.~Ishikawa.
\newblock Note on {F}inslerian relativity.
\newblock {\em J. Math. Phys.}, 22:995--1004, 1981.

\bibitem{javaloyes22b}
M.~{\'A}. Javaloyes and E.~Pend{\'a}s-Recondo.
\newblock Lightlike hypersurfaces and time-minimizing geodesics in cone
  structures.
\newblock In A.~L. et~al. Albujer, editor, {\em Developments in Lorentzian
  Geometry}, pages 159--173, Cham, 2022. {Springer-Verlag}.

\bibitem{javaloyes13}
M.~A. Javaloyes and M.~S\'anchez.
\newblock Finsler metrics and relativistic spacetimes.
\newblock {\em Int. J. Geom. Meth. Mod. Phys.}, 11:1460032, 2014.
\newblock Special Issue for the XXII IFWGP Evora 2-5 Sept.\ 2013.
  {arXiv}:1311.4770.

\bibitem{javaloyes21}
M.~A. Javaloyes, M.~S\'anchez, and F.F. Villase{\~{n}}or.
\newblock The {E}instein-{H}ilbert-{P}alatini formalism in pseudo-{F}insler
  geometry.
\newblock {\em Ad. Theor. Math. Phys.}, 26:3563--3631, 2022.
\newblock {arXiv:}2108.03197.

\bibitem{knebelman29}
M.~S. Knebelman.
\newblock Collineations and {M}otions in {G}eneralized {S}paces.
\newblock {\em Amer. J. Math.}, 51:527--564, 1929.

\bibitem{kupeli87}
D.~N. Kupeli.
\newblock On null submanifolds in spacetimes.
\newblock {\em Geom. Dedicata}, 23:33--51, 1987.

\bibitem{li15}
B.~Li and Z.~Shen.
\newblock Ricci curvature tensor and non-{R}iemannian quantities.
\newblock {\em Canad. Math. Bull.}, 58:530--537, 2015.

\bibitem{li10b}
X.~Li and Z.~Chang.
\newblock Towards a gravitation theory in {B}erwald-{F}insler space.
\newblock {\em Chinese Physics C}, 34:28--34, 2010.

\bibitem{li14}
X.~Li and Z.~Chang.
\newblock Exact solution of vacuum field equation in {F}insler spacetime.
\newblock {\em Phys. Rev. D}, 90:064049, 2014.
\newblock arXiv:1401.6363v1.

\bibitem{lu19}
Y.~Lu, E.~Minguzzi, and S.~Ohta.
\newblock Geometry of weighted {L}orentz-{F}insler manifolds {I}: {S}ingularity
  theorems.
\newblock {\em J. London Math. Soc.}, 104:362--393, 2021.
\newblock {arXiv:}1908.03832.

\bibitem{manjunatha23}
H.~M. Manjunatha, S.~K. Narasimhamurthy, and S.~K. Srivastava.
\newblock Finslerian analogue of the {S}chwarzschild-de {S}itter space-time.
\newblock {\em Pramana - J. Phys.}, 97:90, 2023.

\bibitem{marcal23}
P.\ Mar\c{c}al and Z.~Shen.
\newblock Ricci-flat {F}insler metrics by warped product.
\newblock {\em Proc. Am. Math. Soc.}, 151:2169--2183, 2023.

\bibitem{matsumoto85}
M.~Matsumoto.
\newblock The induced and intrinsic {F}insler connections of a hypersurface and
  {F}inslerien projective geometry.
\newblock {\em J. Math. Kyoto Univ. (JMKYAZ)}, 25:107--144, 1 1985.

\bibitem{matsumoto86}
M.~Matsumoto.
\newblock {\em Foundations of Finsler Geometry and special Finsler Spaces}.
\newblock Kaseisha Press, Tokio, 1986.

\bibitem{michor08}
P.~W. Michor.
\newblock {\em Topics in differential geometry}, volume~93 of {\em Graduate
  Studies in Mathematics}.
\newblock Am. Math. Soc., 2008.

\bibitem{minguzzi14c}
E.~Minguzzi.
\newblock The connections of pseudo-{F}insler spaces.
\newblock {\em Int. J. Geom. Meth. Mod. Phys.}, 11:1460025, 2014.
\newblock Erratum ibid 12 (2015) 1592001. {arXiv}:1405.0645.

\bibitem{minguzzi13d}
E.~Minguzzi.
\newblock Convex neighborhoods for {L}ipschitz connections and sprays.
\newblock {\em Monatsh. Math.}, 177:569--625, 2015.
\newblock {arXiv}:1308.6675.

\bibitem{minguzzi13c}
E.~Minguzzi.
\newblock Light cones in {F}insler spacetime.
\newblock {\em Commun. Math. Phys.}, 334:1529--1551, 2015.
\newblock {{arXiv}:}1403.7060.

\bibitem{minguzzi15}
E.~Minguzzi.
\newblock Raychaudhuri equation and singularity theorems in {F}insler
  spacetimes.
\newblock {\em Class. Quantum Grav.}, 32:185008, 2015.
\newblock {arXiv}:1502.02313.

\bibitem{minguzzi14h}
E.~Minguzzi.
\newblock An equivalence of {F}inslerian relativistic theories.
\newblock {\em Rep. Math. Phys.}, 77:45--55, 2016.
\newblock {arXiv}:1412.4228.

\bibitem{minguzzi15e}
E.~Minguzzi.
\newblock Affine sphere relativity.
\newblock {\em {C}ommun. {M}ath. {P}hys.}, 350:749--801, 2017.
\newblock {arXiv}:1702.06739.

\bibitem{minguzzi15d}
E.~Minguzzi.
\newblock A divergence theorem for pseudo-{F}insler spaces.
\newblock {\em Rep. Math. Phys.}, 80:307--315, 2017.
\newblock {arXiv}:1508.06053.

\bibitem{minguzzi17}
E.~Minguzzi.
\newblock Causality theory for closed cone structures with applications.
\newblock {\em Rev. Math. Phys.}, 31:1930001, 2019.
\newblock {arXiv}:1709.06494.

\bibitem{misner73}
C.~W. Misner, K.~S. Thorne, and J.~A. Wheeler.
\newblock {\em Gravitation}.
\newblock Freeman, San Francisco, 1973.

\bibitem{mo09}
X.~Mo.
\newblock On the non-{R}iemannian quantity {$H$} of a {F}insler metric.
\newblock {\em Differ. Geom. Appl.}, 27:7--14, 2009.

\bibitem{mo06}
Xiaohuan Mo.
\newblock {\em An introduction to Finsler geometry}.
\newblock Peking University, {S}eries in mathematics vol. 1. World
  {S}cientific, New Jersey, 2006.

\bibitem{modugno91}
M.~Modugno.
\newblock Torsion and {R}icci tensor for non-linear connections.
\newblock {\em Diff. Geom. Appl.}, 1:177--192, 1991.

\bibitem{moncrief83}
V.~Moncrief and J.~Isenberg.
\newblock Symmetries of cosmological {C}auchy horizons.
\newblock {\em Comm. Math. Phys.}, 89:387--413, 1983.

\bibitem{moncrief08}
V.~Moncrief and J.~Isenberg.
\newblock Symmetries of higher dimensional black holes.
\newblock {\em Class. Quantum Grav.}, 25:195015, 2008.

\bibitem{moncrief20}
V.~Moncrief and J.~Isenberg.
\newblock Symmetries of cosmological {C}auchy horizons with non-closed orbits.
\newblock {\em Comm. Math. Phys.}, 374:145--186, 2020.

\bibitem{ohta15}
S.~Ohta.
\newblock Splitting theorems for {F}insler manifolds of nonnegative {R}icci
  curvature.
\newblock {\em Reine Angew. Math.}, 700:155--174, 2015.

\bibitem{ohta21}
S.~Ohta.
\newblock {\em Comparison {F}insler geometry}.
\newblock Springer Nature, Cham, 2021.

\bibitem{petersen18b}
O.~L. Petersen.
\newblock Wave equations with initial data on compact {C}auchy horizons.
\newblock {\em Analysis \& PDE}, 14:2363--2408, 2021.
\newblock {arXiv:}1802.10057.

\bibitem{petersen18}
O.~L. Petersen and I.~R{\'a}cz.
\newblock Symmetries of vacuum spacetimes with a compact {C}auchy horizon of
  constant non-zero surface gravity.
\newblock {\em Ann. Henri Poincar\'e}, 24, 2023.
\newblock {arXiv:}1809.02580.

\bibitem{pfeifer12}
C.~Pfeifer and M.~N.~R. Wohlfarth.
\newblock Finsler geometric extension of {E}instein gravity.
\newblock {\em Phys. Rev. D}, 85:064009, 2012.

\bibitem{rademacher04}
H.~Rademacher.
\newblock {\em Nonreversible Finsler metrics of positive flag curvature},
  volume A sampler of Riemann-Finsler geometry, volume 50 of Math. Sci. Res.
  Inst. Publ., pages 261--302.
\newblock Cambridge Univ. Press, Cambridge, 2004.

\bibitem{reiris21}
M.~Reiris and I.~Bustamante.
\newblock On the existence of {K}illing fields in smooth spacetimes with a
  compact {C}auchy horizon.
\newblock {\em Class. Quantum Grav.}, 38:075010, 2021.

\bibitem{rund59}
H.~Rund.
\newblock {\em The differential geometry of {F}insler spaces}.
\newblock {Springer-Verlag}, Berlin, 1959.

\bibitem{rund62}
H.~Rund.
\newblock \"{U}ber {F}inslersche {R}aume mit speziellen
  {K}r\"ummungseigenschaften.
\newblock {\em Monatsh. Math.}, 66:241--251, 1962.

\bibitem{rutz93}
S.~F. Rutz.
\newblock A {F}insler generalisation of {E}instein's vacuum field equations.
\newblock {\em Gen. Relativ. Gravit.}, 25:1139--1158, 1993.

\bibitem{saveliev12}
N.~Saveliev.
\newblock {\em Lectures on the Topology of 3-Manifolds}.
\newblock De Gruyter, Berlin, 2012.

\bibitem{sevim23}
E.~S. Sevim, Z.~Shen, and S.~Ulgen.
\newblock On some {R}icci curvature tensors in {F}insler geometry.
\newblock {\em Mediterr. J. Math.}, 20:231, 2023.

\bibitem{shen01}
Z.~Shen.
\newblock {\em Lectures on {F}insler geometry}.
\newblock World {S}cientific, Singapore, 2001.

\bibitem{shibata78}
C.~Shibata.
\newblock On the curvature tensor {$R_{hijk}$} of {F}insler spaces of scalar
  curvature.
\newblock {\em Tensor N. S.}, 32:311--317, 1978.

\bibitem{szilasi14}
J.~Szilasi, R.~L. Lovas, and D.~Cs. Kertesz.
\newblock {\em Connections, sprays and Finsler structures}.
\newblock World {S}cientific, London, 2014.

\bibitem{takano74}
Y.~Takano.
\newblock Gravitational field in {F}insler spaces.
\newblock {\em Lettere al Nuovo Cimento}, 10:747--750, 1974.

\bibitem{takano88}
Y.~Takano.
\newblock {\em Introduction to the Theory of Fields in Finsler Spaces}, volume
  Gravitational Measurements, Fundamental Metrology and Constants, De Sabbata,
  V., Melnikov, V.N. (eds) of {\em NATO ASI Series, vol 230}, pages 459--466.
\newblock Springer, Dordrecht., 1988.

\bibitem{vattamany04}
Sz. Vattam{\'a}ny.
\newblock {\em On the projective geometry and metrizability of spray
  manifolds}.
\newblock PhD thesis, Debrecen University, 2004.

\end{thebibliography}

\end{document}